%% file: 2019-SPAA-Space-Time.tex
\documentclass[sigconf]{acmart}

\usepackage{balance}

\def\BibTeX{{\rm B\kern-.05em{\sc i\kern-.025em b}\kern-.08emT\kern-.1667em\lower.7ex\hbox{E}\kern-.125emX}}

\input{macros.tex}

\copyrightyear{2019} 
\acmYear{2019} 
\setcopyright{acmcopyright}
\acmConference[SPAA '19]{31st ACM Symposium on Parallelism in Algorithms and Architectures}{June 22--24, 2019}{Phoenix, AZ, USA}
\acmBooktitle{31st ACM Symposium on Parallelism in Algorithms and Architectures (SPAA '19), June 22--24, 2019, Phoenix, AZ, USA}
\acmPrice{15.00}
\acmDOI{10.1145/3323165.3323209}
\acmISBN{978-1-4503-6184-2/19/06}

\begin{document}
\title[Data Races and the Discrete Resource-time Tradeoff Problem]{Data Races and the Discrete Resource-time Tradeoff Problem with Resource Reuse over Paths}

\author{Rathish Das}
\affiliation{%
	\institution{Stony Brook University}
}
\email{radas@cs.stonybrook.edu}

\author{Shih-Yu Tsai}
\affiliation{%
	\institution{Stony Brook University}
}
\email{shitsai@cs.stonybrook.edu}

\author{Sharmila Duppala}
\affiliation{%
	\institution{Stony Brook University}
}
\email{sduppala@cs.stonybrook.edu}

\author{Jayson Lynch}
\affiliation{%
	\institution{MIT}
}
\email{jaysonl@mit.edu}

\author{Esther M. Arkin}
\affiliation{%
	\institution{Stony Brook University}
}
\email{esther.arkin@stonybrook.edu}

\author{Rezaul Chowdhury}
\affiliation{%
	\institution{Stony Brook University}
}
\email{rezaul@cs.stonybrook.edu}

\author{Joseph S. B. Mitchell}
\affiliation{%
	\institution{Stony Brook University}
}
\email{joseph.mitchell@stonybrook.edu}

\author{Steven Skiena}
\affiliation{%
	\institution{Stony Brook University}
}
\email{skiena@cs.stonybrook.edu}

\renewcommand{\shortauthors}{Das, Tsai, Duppala, Lynch, Arkin, Chowdhury, Mitchell, Skiena}

\begin{abstract}
\input{abstract}
\end{abstract}

%
%



\maketitle

\begin{table}[h!]
	\centering
	{\Large {\textbf{\textsc{\underline{Page Distribution}}}}}~\\~\\
	\scalebox{1.4}[1.2]{
		{\small 
		\hspace{-0.25cm}	\begin{colortabular}{ | l | c | c |}
				\hline                       
				\rowcolor{tabletitlecolor} Page Type & Page Numbers & \#Pages \\  \hline

				Title & 1 & 1\\
				
				\rowcolor{tablealtrowcolor} Main Text & 2, 4, 5, 7 -- 10, 12, 15, 16  & 10 \\
				
				Figures & 3, 6, 11, 13, 14, 17  & 6\\								

				\rowcolor{tablealtrowcolor} Bibliography & 18  & 1 \\

				Appendices & 19 -- 20 & 2\\								

				\hline
			\end{colortabular}
		}
	}
		\vspace{0.3cm}
		\label{tab:org}
\end{table}

\newpage

\input{introduction}
\input{related_work}
\input{problem_formulation}

\input{approx}
\input{binary_splitting_approx}
\input{exact_algorithm}
\input{hardness}

\input{treewidth}

\input{conclusion}


\clearpage
\newpage
\bibliographystyle{ACM-Reference-Format}
\bibliography{tradeoff}

\clearpage
\newpage
\begin{appendix}

\begin{center}
{\Huge{\sc Appendix}}
\end{center}

\input{threeDMatching_hardness}

\end{appendix}

\end{document}

%% file: macros.tex
\usepackage{amsmath,amssymb,amsthm}
\usepackage{graphicx}
\usepackage[mathscr]{euscript}
\usepackage{tikz}
\usetikzlibrary{trees,arrows,automata}
\usetikzlibrary{shapes.geometric} 
\usetikzlibrary{shadows,fadings}
\usepackage{wrapfig}   
\usepackage{stmaryrd} 
\usepackage{url}
\usepackage{pifont}
\usepackage{lipsum}
\usetikzlibrary{arrows,shapes,positioning,shadows,trees}
\usepackage{multirow}
\usepackage{tcolorbox}
\tcbuselibrary{skins}
\usepackage{environ}
\usepackage{colortbl}
\usepackage{hhline}
\usepackage{rotating}
\usepackage{array}
\usepackage{libertine}
\usepackage{enumitem}
\usepackage{mdframed}

\newcommand{\hide}[1]{}

\newtheorem{observation}{Observation}[section]
\newtheorem{question}{Question}[section]

\newcommand{\floor}[1]      {\left\lfloor #1 \right\rfloor}

\definecolor{gray}{rgb}{0.3,0.3,0.3}
\definecolor{orange}{rgb}{1,0.3,0}

\def\func#1{\textrm{\bf{\sc{#1}}}}

\newcommand{\m}[1]{\mathscr{#1}}

\newcommand{\para}[1]{\vspace{0.1cm}\noindent{\bf{\boldmath{#1}}}}

\newcommand{\codesize}{\scriptsize}

\newcommand{\T}{\hspace{1em}}

\def\func#1{\textrm{\bf{\sc{#1}}}}

\newcommand{\Oh}[1]{{\mathcal O}\left({#1}\right)}

\newcommand{\Th}[1]{{\Theta}\left({#1}\right)}

\newcommand{\xfor}{{\bf{{for~}}}}

\newcommand{\xto}{{\bf{{to~}}}}

\newcommand{\xdo}{{\bf{{do~}}}}

\newcommand{\xparallelfor}{{\bf{{parallel~for~}}}}

\newcommand{\vsitem}{\vspace{-0.15cm}\item}

\makeatletter
\newcommand{\ALOOP}[1]{\ALC@it\algorithmicloop\ #1%
  \begin{ALC@loop}}
\newcommand{\ENDALOOP}{\end{ALC@loop}\ALC@it\algorithmicendloop}

\makeatother

\definecolor{gray}{rgb}{0.3,0.3,0.3}
\colorlet{lightblue}{blue!15}
\colorlet{lightred}{red!20}
\colorlet{lightyellow}{green!20}
\colorlet{framecolor}{yellow!20}
\colorlet{frameloopcolor}{yellow!20}
\colorlet{algocolor}{orange!20}

\surroundwithmdframed[
  backgroundcolor   = algocolor,
  hidealllines      = true,
  innerleftmargin   = 0.5em,
  innerrightmargin  = 0.5em,
  innertopmargin    = -1.1em,
  innerbottommargin = 0em,
  leftline=true,
  rightline=true,
  bottomline=true,
  topline=true,
  skipabove         = .5\baselineskip,
  skipbelow         = .5\baselineskip
]{algorithm}

\makeatletter
   \newcommand\figcaption{\def\@captype{figure}\caption}
   \newcommand\tabcaption{\def\@captype{table}\caption}
\makeatother

\colorlet{algotitlebarcolor}{gray!20}
\colorlet{algotitlebarcolortop}{gray!40}
\colorlet{algotitlebarcolorbottom}{gray!20}
\colorlet{framecolor}{white!20}
\colorlet{framecolortop}{white!40}
\colorlet{framecolorbottom}{white!20}
\colorlet{bluetop}{cyan!40}
\colorlet{bluebottom}{cyan!20}
\colorlet{greentop}{green!40}
\colorlet{greenbottom}{green!20}
\colorlet{shadowcolor}{gray}
\colorlet{algotitlecolor}{black}
\colorlet{algoframecolor}{gray}
\colorlet{tabletitlecolor}{gray!30}
\colorlet{tablealtrowcolor}{gray!15}
\colorlet{tablefillcolor}{white!40}
\colorlet{tablefillcolor2}{white!20}

\newtcolorbox{mycolorbox}[1]
{skin=enhanced,colbacktitle=algotitlebarcolor,coltitle=algotitlecolor,colback=algocolor,colframe=algoframecolor,title style={top color=algotitlebarcolortop, bottom color=algotitlebarcolorbottom},interior style={top color=framecolortop, bottom color=framecolorbottom},drop shadow,title={\codesize #1}}

\newcommand{\algotopspace}{\vspace{-0.35cm}}

\newcommand{\algobottomspace}{\vspace{-0.15cm}}

\arrayrulecolor{gray}
\newsavebox{\tablebox}
\NewEnviron{colortabular}[1]{%
  \addtolength{\extrarowheight}{1ex}%
  \savebox{\tablebox}{%
    \begin{tabular}{#1}%
	  \rowcolor{tabletitlecolor}
      \BODY%
    \end{tabular}}
  \begin{tikzpicture}
    \begin{scope}
      \clip[rounded corners=1ex] (0,-\dp\tablebox) -- (\wd\tablebox,-\dp\tablebox) -- (\wd\tablebox,\ht\tablebox) -- (0,\ht\tablebox) -- cycle;
      \node at (0,-\dp\tablebox) [anchor=south west,inner sep=0pt]{\usebox{\tablebox}};
    \end{scope}
    \draw[rounded corners=1ex] (0,-\dp\tablebox) -- (\wd\tablebox,-\dp\tablebox) -- (\wd\tablebox,\ht\tablebox) -- (0,\ht\tablebox) -- cycle;
  \end{tikzpicture}
}

%% file: abstract.tex
A determinacy race occurs if two or more logically parallel instructions access the same memory location and at least one of them tries to modify its content. Races are often undesirable as they can lead to nondeterministic and incorrect program behavior. A data race is a special case of a determinacy race which can be eliminated by associating a mutual-exclusion lock with the memory location in question or allowing atomic accesses to it. However, such solutions can reduce parallelism by serializing all accesses to that location. For associative and commutative updates to a memory cell, one can instead use a reducer, which allows parallel race-free updates at the expense of using some extra space. More extra space usually leads to more parallel updates, which in turn contributes to potentially lowering the overall execution time of the program. 

We start by asking the following question. Given a fixed budget of extra space for mitigating the cost of races in a parallel program, which memory locations should be assigned reducers and how should the space be distributed among those reducers in order to minimize the overall running time? We argue that under reasonable conditions the races of a program can be captured by a directed acyclic graph (DAG), with nodes representing memory cells and arcs representing read-write dependencies between cells. We then formulate our original question as an optimization problem on this DAG. We concentrate on a variation of this problem where space reuse among reducers is allowed by routing every unit of extra space along a (possibly different) source to sink path of the DAG and using it in the construction of multiple (possibly zero) reducers along the path. We consider two different ways of constructing a reducer and the corresponding duration functions (i.e., reduction time as a function of space budget). 

We generalize our race-avoiding space-time tradeoff problem to a discrete resource-time tradeoff problem with general non-increasing duration functions and resource reuse over paths of the given DAG. 

For general DAGs, we show that even if the entire DAG is available to us offline the problem is strongly NP-hard under all three duration functions, and we give approximation algorithms for solving the corresponding optimization problems. We also prove hardness of approximation for the general resource-time tradeoff problem and give a pseudo-polynomial time algorithm for series-parallel DAGs.

%% file: introduction.tex
\section{Introduction}


A determinacy race (or a general race) \cite{netzer1992race,feng1999efficient} occurs if two or more logically parallel instructions access the same memory location and at least one of them modifies its content. Races are often undesirable as they can lead to nondeterministic and incorrect program behavior. A data race is a special case of a determinacy race which can be eliminated by associating a mutual-exclusion lock with the memory location in question or allowing only atomic accesses to it. 
%
%
Such a solution, however, makes all accesses to that location serial and thus destroys all parallelism. Figure \ref{fig:race} shows an example.

\begin{figure*}[h!]
	\centering
	\begin{minipage}{\textwidth}
		\begin{minipage}{0.98\textwidth}
			\vspace{-0.2cm}
			\centering
\part{\label{title}}\includegraphics[width=\textwidth]{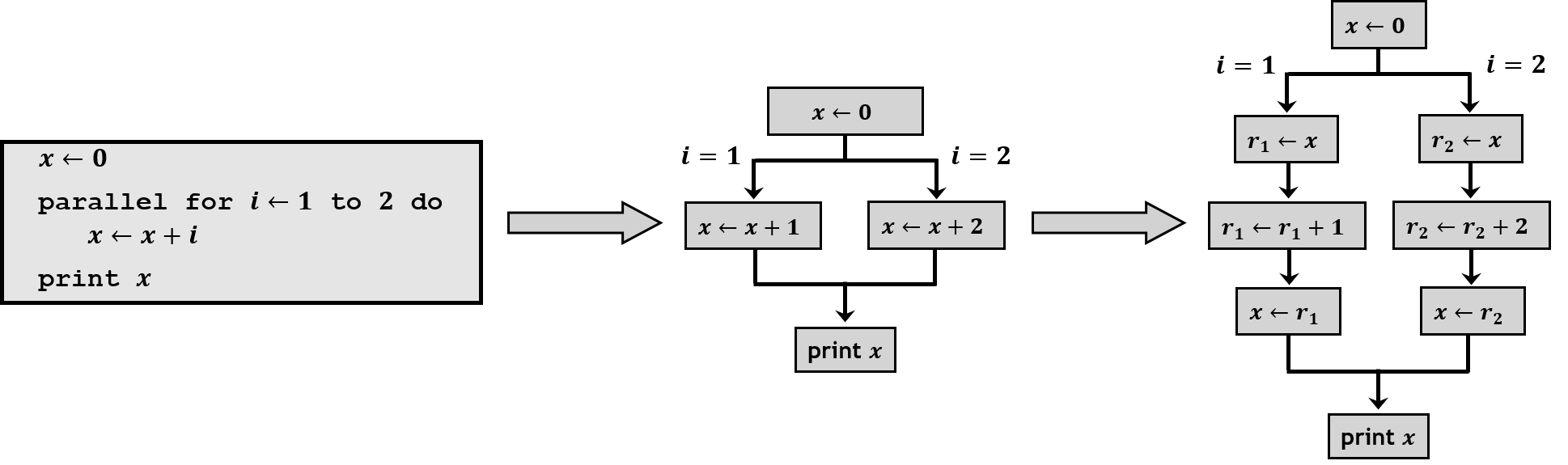}
          		\caption{This figure shows a race on global variable $x$ caused by two parallel threads trying to increment $x$, where $r_1$ and $r_2$ are local registers. The value printed by the `print' statement depends on how the two threads are scheduled. Unless the two threads are executed sequentially, the print statement will print an incorrect result (either $1$ or $2$ depending on which thread updated $x$ last).}
	        	\label{fig:race}
		\end{minipage}
	\end{minipage}
	\vspace{-0.1cm}
\end{figure*}

One can use a reducer \cite{frigo2009reducers,openmp1997openmp,reinders2007intel} to eliminate data races on a shared variable without destroying parallelism, provided the update operation is associative and commutative. Figure \ref{fig:binary-reducer} shows the construction of a simple recursive binary reducer. For any integer $h > 0$ such a reducer is a full binary tree of height $h$ and size $2^{h + 1} - 1$ with the shared variable at the root. Each nonroot node is associated with a unit of extra space initialized to zero. All updates to the shared variable are equally distributed among the leaves of the tree. Each node has a lock and a waiting queue to avoid races by serializing the updates it receives, but updates to different nodes can be applied in parallel. As soon as a node undergoes its last update, it updates its parent using its final value. In fact, such a reducer can be constructed using only $2^h$ units of extra space because if a node completes before its sibling it can become its own parent (with ties broken arbitrarily) and the sibling then updates the new parent. Assume that the time needed to apply an update significantly dominates the execution time of every other operation the reducer performs and each update takes one unit of time to apply. Then a reducer of height $h$ can correctly apply $n$ parallel updates on a shared variable in $\lceil{\frac{n}{2^h}}\rceil + h + 1$ time provided at least $2^h$ processors are available. Hence, for large $n$, the speedup achieved by a reducer (w.r.t. serially and directly updating the shared variable) is almost linear in the amount of extra space used.

\begin{figure*}[t!]
	\centering
	\begin{minipage}{\textwidth}
		\begin{minipage}{0.5\textwidth}
			\vspace{-0.2cm}
			\centering
			\includegraphics[width=\textwidth]{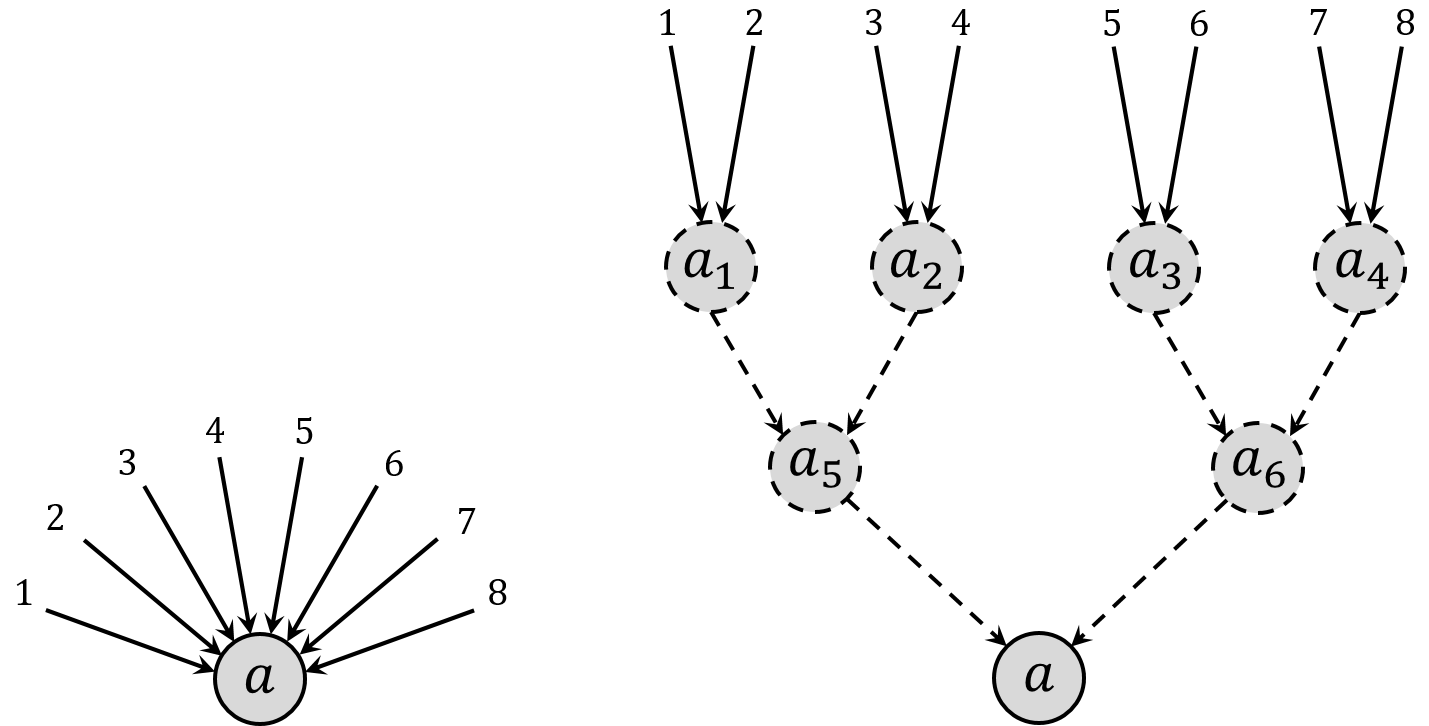}
          		\caption{{{\textsc{\textbf{[Left]}}}} A memory location $a$ with eight updates using an associative and commutative operator. 
                                 {{\textbf{\textsc{[Right]}}}} The same location $a$ with a recursive binary reducer of height two on top of it.}
	        	\label{fig:binary-reducer}
		\end{minipage}
		\begin{minipage}{0.03\textwidth}
                 ~
		\end{minipage}
		\begin{minipage}{0.40\textwidth}
			\begin{mycolorbox}{\func{Parallel-MM}$( Z, X, Y, n )$}
				\begin{minipage}{0.99\textwidth}
					{\codesize
						\algotopspace{}
						\noindent
							\vspace{0.3cm}
						\begin{enumerate}
							
							\setlength{\itemindent}{-2em}
							
							\vsitem \xparallelfor $i \gets 1$ \xto $n$ \xdo \label{ln:par-MM-i}
							\vspace{0.1cm}
							\vsitem \T \xparallelfor $j \gets 1$ \xto $n$ \xdo \label{ln:par-MM-j}
								\vspace{0.1cm}
							\vsitem \T\T $Z[ i ][ j ] \gets 0$ \label{ln:par-MM-init}
						\vspace{0.1cm}
							\vsitem \T\T \xfor $k \gets 1$ \xto $n$ \xdo \label{ln:par-MM-k}
							\vspace{0.1cm}
							\vsitem \T\T\T $Z[ i ][ j ] \gets Z[ i ][ j ] + X[ i ][ k ] \times Y[ k ][ j ]$ \label{ln:par-MM-update}
							
							\algobottomspace{}
						\end{enumerate}
					}
				\end{minipage}
			\end{mycolorbox}
       			\vspace{0.05cm}
			\caption{Parallel code that multiplies two $n \times n$ matrices $X[1..n][1..n]$ and $Y[1..n][1..n]$, and puts the result in $Z[1..n][1..n]$.}
			\label{fig:par-MM}
		\end{minipage}
	\end{minipage}
	\vspace{-0.1cm}
\end{figure*}

To see how extra space can speed up real parallel programs consider the iterative matrix multiplication code \func{Parallel-MM} shown in Figure \ref{fig:par-MM} which multiplies two $n \times n$ matrices $X[1..n][1..n]$ and $Y[1..n][1..n]$ and puts the results in another $n \times n$ matrix $Z[1..n][1..n]$; that is, it sets $Z[i][j] = \sum_{1 \leq k \leq n}{X[i][k] \times Y[k][j]}$ for $1 \leq i, j \leq n$. Since every $Z[i][j]$ value can be computed independently of others, all iterations of the loops in Lines \ref{ln:par-MM-i} and \ref{ln:par-MM-j} can be executed in parallel without compromising correctness of the computation. However, the same is not true for the loop in Line \ref{ln:par-MM-k} because if parallelized, for fixed values of $i$ and $j$, all iterations of that loop will update the same memory location $Z[i][j]$ giving rise to data races and thus producing potentially incorrect results. 
\hide{ 
Using a mutual-exclusion lock or atomic updates for each $Z[i][j]$ will ensure correctness but at the expense of destroying all parallelism of that loop which is asymptotically the same as not parallelizing the loop at all. Thus even with an unbounded number of processors, the code will take $\Th{n}$ time to multiply the two $n \times n$ matrices. 
}%
%
Use of a mutual-exclusion lock or atomic updates for each $Z[i][j]$ will ensure correctness but in that case even with an unbounded number of processors, the code will take $\Th{n}$ time to multiply the two $n \times n$ matrices.
Now if we put a reducer of height $h$ (integer $h \in [1, \log_2{n}]$) at the top of each $Z[i][j]$ the time to fully update each $Z[i][j]$ and thus the overall running time of the code will drop to $\Th{\frac{n}{2^h} + h}$ at the cost of using $n^{2} \times 2^{h}$ units of extra space. Observe that when $h = 1$, the running time of the code almost halves using $2n^2$ units of extra space, and when $h = \floor{\log_{2}{n}}$, the running time drops to $\Th{\log{n}}$ using $\Th{n^3}$ extra space.

\begin{figure*}[t!]
	\centering
	\begin{minipage}{\textwidth}
		\begin{minipage}{0.48\textwidth}
			\vspace{-0.2cm}
			\centering
			\includegraphics[width=\textwidth]{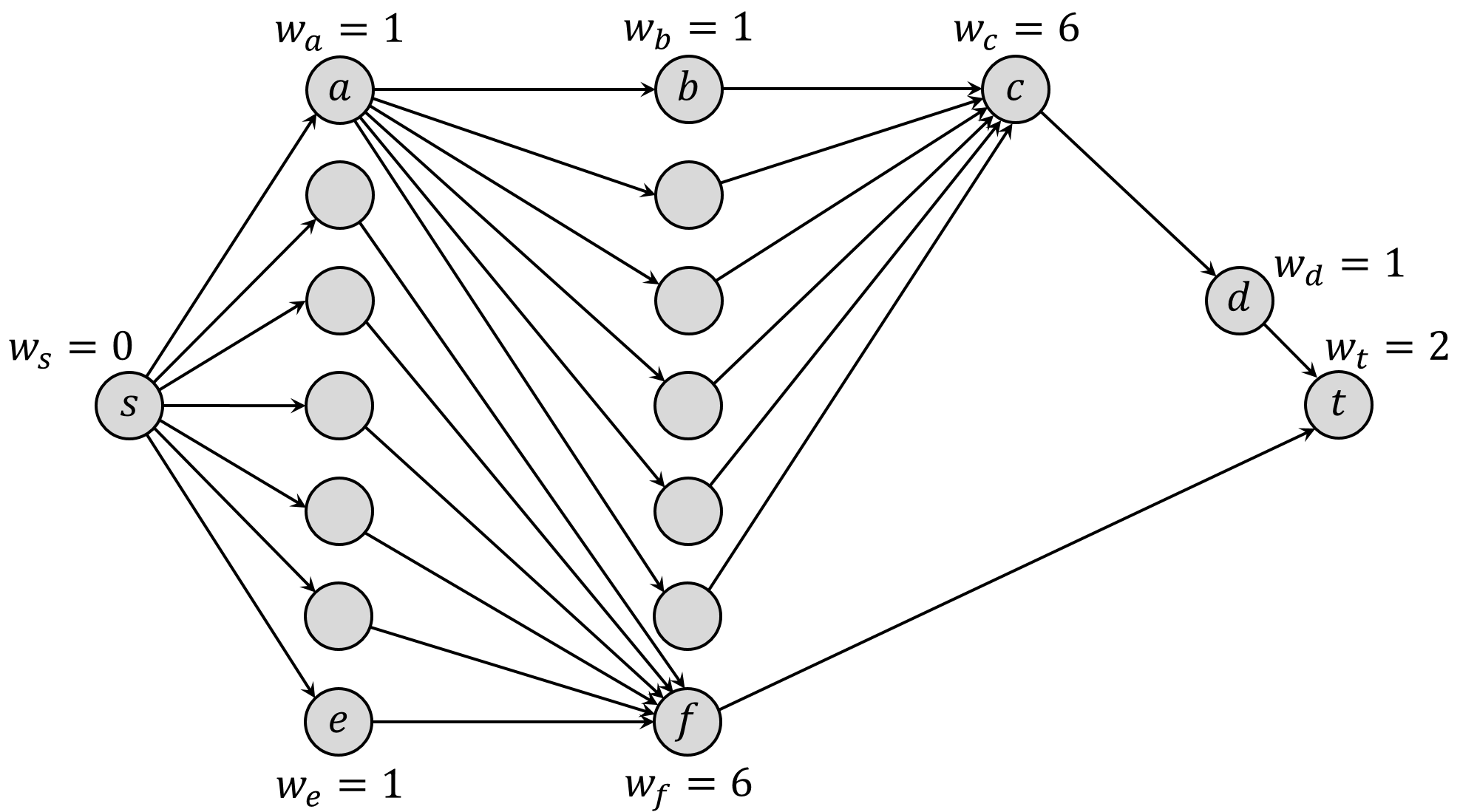}
          		\caption{A DAG in which each node's work value is set to its in-degree. The makespan of this DAG is 11, and
path $s \rightarrow a \rightarrow b \rightarrow c \rightarrow d \rightarrow t$ achieves it.\vspace{1cm}} 
	        	\label{fig:example-DAG}
		\end{minipage}
		\begin{minipage}{0.03\textwidth}
                 ~
		\end{minipage}
		\begin{minipage}{0.48\textwidth}
			\vspace{-0.2cm}
			\centering
			\includegraphics[width=\textwidth]{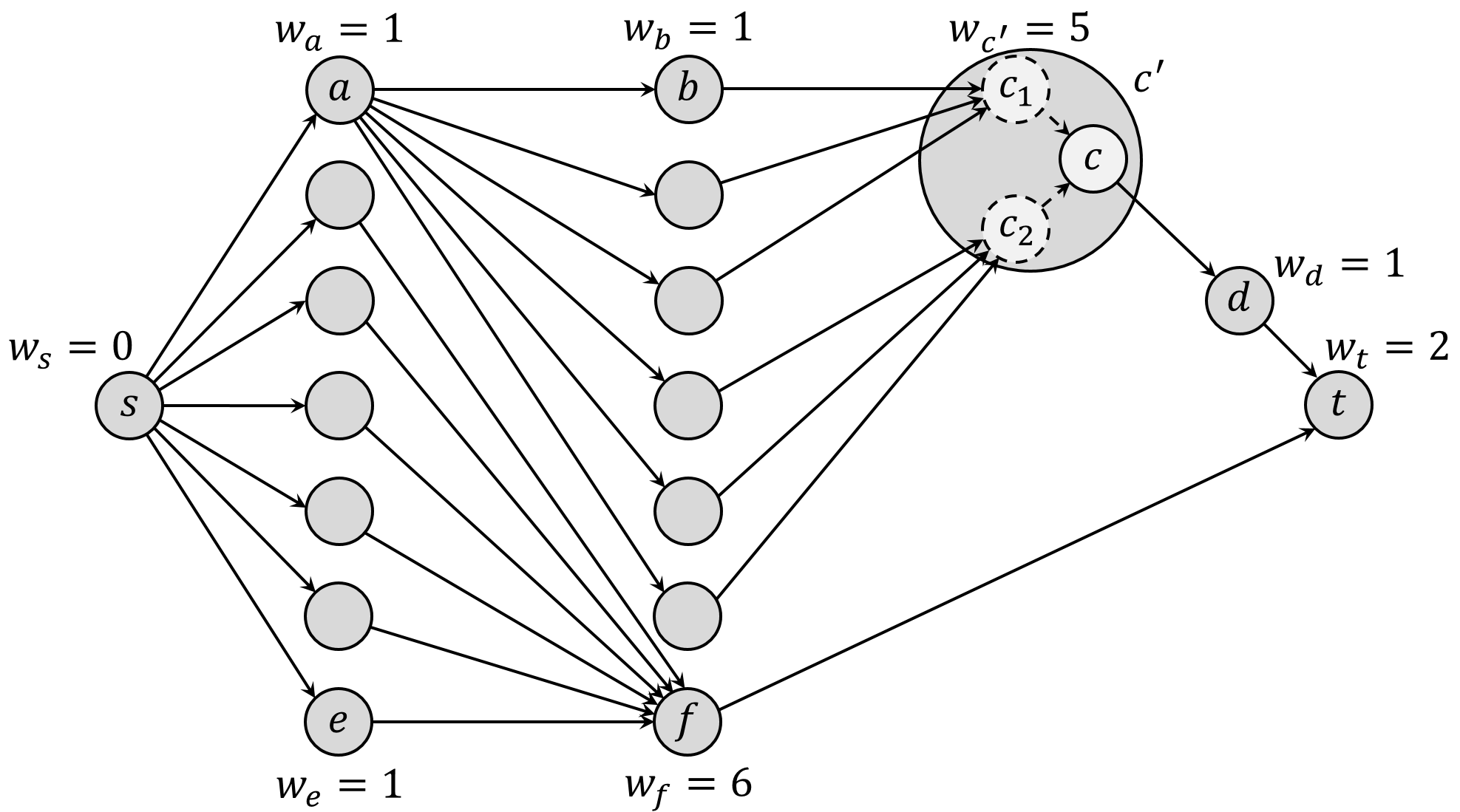}
          		\caption{Node $c$ from the DAG in Figure \ref{fig:example-DAG} has been replaced with a supernode $c'$ in this figure which is nothing but node $c$ with a reducer of height $1$ on top. The makespan of this reduced DAG is 10, and
path $s \rightarrow a \rightarrow b \rightarrow c_1 \rightarrow c \rightarrow d \rightarrow t$ achieves it.}
	        	\label{fig:example-DAG-with-a-reducer}
		\end{minipage}
	\end{minipage}
	\vspace{-0.1cm}
\end{figure*}

\hide{
Unlike the \func{Parallel-MM} code of Figure \ref{fig:par-MM}, many parallel programs, such as dynamic programs, programs solving graph problems, and molecular dynamics codes have complicated and/or irregular structures, and determining how to use a given amount of extra space to minimize their running times by constructing reducers is not a straightforward task. In order to answer such questions for any given program we capture the races of the program in a DAG. 
}
%

In order to analyze a program $P$ with data races, we capture those races in a directed acyclic graph (DAG) $D(P)$, assuming that there are no cyclic read-write dependencies among the memory locations accessed by $P$. Figure \ref{fig:example-DAG} shows an example. We restrict $P$ to the set of programs that perform $\Oh{1}$ other operations between two successive writes to the memory, e.g., \func{Parallel-MM} in Figure \ref{fig:par-MM}. We assume that an update operation is significantly more expensive than any other single operation performed by $P$ and hence the costs of those operations can be safely ignored.
Each node $x$ of $D(P)$ represents a memory location, and a directed edge from node $x$ to node $y$ means that $y$ is updated using the value stored at $x$. The in-degree $d^{(in)}_x$ of node $x$ gives the number of times $x$ is updated. With $x$ we also associate a {\em work} value $w_x$ and set $w_x = d^{(in)}_x$. Assuming that each update operation requires unit time to execute and each node has a lock and a wait queue to serialize the updates, the $w_x$ value represents the time spent updating $x$ (excluding all idle times). The $w_x$ value also represents an upper bound on the time elapsed between the trigger time of any incoming edge of $x$ and the time the edge completes updating $x$. We assume that updates along all outgoing edges of $x$ trigger as soon as all incoming edges complete updating $x$. One can then make the following observation.
\begin{observation}\label{obs:makespan} 
The running time of $P$ with an unbounded number of processors is upper bounded by the makespan of $D(P)$\footnote{To see why this is true start from the sink node and move backward toward the source by always moving to that predecessor $y$ of the current node $x$ that performed the last update on $x$ and noting that after edge $(y, x)$ was triggered it did not have to wait for more than $d^{(in)}_x$ time units to complete applying $y$'s update to $x$.}.
\end{observation} 

Then one natural question to ask is the following.

\begin{question}\label{ques:nonreuse}
Given a fixed budget of units of extra space to mitigate the cost of data races in $P$, which memory locations should be assigned reducers and how
should the space be distributed among those reducers in order to minimize the makespan of $D(P)$?
\end{question}
 
Figure \ref{fig:example-DAG-with-a-reducer} shows how to minimize the makespan of the DAG in Figure \ref{fig:example-DAG} using two units of extra space.

The question above ignores the possibility that space can be reused among reducers in $D(P)$. Indeed, after node $x$ reaches its final value (i.e., updated $w_x = d^{(in)}_x$ times) it can release all (if any) space it used for its reducer which can then be reused by some other node $y$. A global memory manager can be used by the nodes to allocate/deallocate space for reducers. The following modified version of Question \ref{ques:nonreuse} now allows space reuse.

\begin{question}\label{ques:reuse-global}
Repeat Question \ref{ques:nonreuse} but allow for space reuse among nodes of $D(P)$ by putting all extra space under the control of a global memory manager that each node calls to allocate space for its reducer right before its first update and to deallocate that space right after its last update.
\end{question}

The problem with a single global memory manager is that it can easily become a performance bottleneck for highly parallel programs. Though better memory allocators have been developed for multi-core or multi-threaded systems \cite{aigner2015fast, berger2000hoard, llalloc, schneider2006scalable, TCMalloc}, we can instead use an approach often used by recursive fork-join programs which avoids repeated calls to an external memory manager altogether along with the overhead of repeated memory allocations/deallocations. 
\hide{
There are memory allocators based on global memory manager for multi-core or multi-threaded systems such as scalloc \cite{aigner2015fast}, Hoard \cite{berger2000hoard}, llalloc \cite{llalloc}, Streamflow \cite{schneider2006scalable}, and TCMalloc \cite{TCMalloc}. They use thread-local space for memory allocation and a global manager for memory deallocation/reuse. For the global manager, they use concurrent data structures. However, these data structures can not completely avoid the need for synchronization \cite{aigner2015fast, henzinger2013quantitative, shavit2011data} without compromising correctness.   
Instead, we can use an approach often used by recursive fork-join programs to avoid the overhead of synchronization and repeated memory allocations/deallocations.
}
%
%
A single large segment of memory is allocated before the initial recursive call is made and a pointer to that segment is passed to the recursive call. Each recursive call splits and distributes its segment among its child recursive calls and reclaims the space when the children complete execution. So, we will assume that all the given extra space initially reside at the source node (i.e., node with in-degree zero). Then they flow along the edges toward the sink node (i.e., node with outdegree zero) possibly splitting along outgoing edges and merging at the tip of incoming edges as they flow. Each unit of space reaching node $x$ moves out of $x$ along some outgoing edge as soon as $x$ becomes fully updated and those edges trigger. Every unit of space may participate in the construction of multiple reducers (possibly zero) along the path it takes.

\begin{question}\label{ques:reuse-flow}
Repeat Question \ref{ques:nonreuse} but now allow for space reuse among nodes of $D(P)$ by flowing each unit of space along a source to sink path and using it in the construction of zero or more reducers along that path.
\end{question}

While several existing results \cite{de1997complexity,du1989complexity,skutella1998approximation,jansen2006approximation} can be extended to answer Questions \ref{ques:nonreuse} and \ref{ques:reuse-global}, to the best of our knowledge, Question \ref{ques:reuse-flow} had not been raised before. In this paper we investigate answers to Question \ref{ques:reuse-flow} by extending it to a more general resource-time tradeoff question posed on a DAG in which nodes represent jobs (not necessarily of updating memory locations), resources (not necessarily space) flow along source to sink paths, and an general duration function (i.e., time needed to complete a job as a function of the amount of resources used) is specified for each node. We consider the following three duration functions: general non-increasing function for the general resource-time question, and recursive binary reduction and multiway ($k$-way) splitting for the space-time case. 

\hide{
We analyze all these variants in terms of their hardness and give approximation algorithms for solving the corresponding optimization problems. We also prove hardness of approximation for the general general non-increasing duration function, and give a pseudo-polynomial time algorithm for series-parallel DAGs. 
}

For general DAGs, we show that even if the entire DAG is available to us offline the problem is strongly NP-hard under all three duration functions, and we give approximation algorithms for solving the corresponding optimization problems. We also prove hardness of approximation for the general resource-time tradeoff problem and give a pseudo-polynomial time algorithm for series-parallel DAGs. Our main results are summarized in Table \ref{tab:results}. 

\hide{
\begin{tabular}{ |p{3cm}||p{3cm}|p{3cm}|p{3cm}|  }
	\hline
	\multicolumn{4}{|c|}{ \textbf{Results} } \\
	\hline
	\textbf{Duration Function} & \textbf{Hardness} & \textbf{Hardness of approximation}& \textbf{Approximation Ratio}\\
	\hline
	General non-increasing function   & Strongly NP-hard &Makespan minimization: $<2$ Resource minimization: $<\frac{3}{2}$ &   
	($\frac{1}{\alpha}$,$\frac{1}{1-\alpha}$) bi-criteria, $0<\alpha<1$\\
	\hline
	Binary Recursive &   Strongly NP-hard  & - & Makespan minimization: $4$ and  ($\frac{4}{3},\frac{14}{5}$) bi-criteria \\
	
	\hline
	$k$-way Splitting &   Strongly NP-hard  & - & Makespan minimization: $5$  \\
	\hline
\end{tabular}
}

\begin{table*}[ht!]
	\centering
	\scalebox{0.88}[0.88]{
		{\small 
		\hspace{-0.25cm}	\begin{colortabular}{ | l | c | c | c |}
				\hline                       
				\rowcolor{tabletitlecolor} Duration function & Hardness & Hardness of Approximation & Approximation Results \\  \hline

				General non-increasing & \begin{tabular}{@{}c@{}} strongly NP-hard\end{tabular} & \begin{tabular}{@{}l@{}} $\bullet$ makespan $< 2~\textrm{OPT}$ with resources fixed\\$\bullet$ resource $< \frac{3}{2}~\textrm{OPT}$ with makespan fixed\end{tabular} & \begin{tabular}{@{}c@{}} $\left(\frac{1}{\alpha}, \frac{1}{1-\alpha}\right)$ bi-criteria (resource, makespan),\\$0<\alpha<1$\end{tabular}\\
				
			\rowcolor{tablealtrowcolor}	Recursive binary & \begin{tabular}{@{}c@{}} strongly NP-hard\end{tabular} & -- & \begin{tabular}{@{}l@{}} $\bullet$ makespan $\leq 4~\textrm{OPT}$ with resources fixed\\$\bullet$ $\left(\frac{4}{3}, \frac{14}{5}\right)$ bi-criteria (resource, makespan)\end{tabular} \\
				
				Multiway splitting & \begin{tabular}{@{}c@{}} strongly NP-hard\end{tabular} & -- & \begin{tabular}{@{}l@{}} makespan $\leq 5~\textrm{OPT}$ with resources fixed\end{tabular}\\								
				\hline
			\end{colortabular}
		}
	}
		\vspace{0.3cm}
	\caption{{ Our main results on resource-time tradeoff problems in which resources are routed along source to sink paths (i.e., related to Question \ref{ques:reuse-flow} and its generalization).}}
		\vspace{-0.7cm}
		\label{tab:results}
		\end{table*}

\hide{
{\bf[BEGIN OLD TEXT]}

Space plays an important role in optimizing the parallelism thereby reducing the time needed to complete a computational task in a parallel environment. Although we have enough processors to have a high parallelism, it is often the amount of space that becomes the bottleneck on how much parallelism can be exploited, when we have data transfer instructions that involve memory(space). Space constraints on the system decides the parallelism in such cases. For example, if we have multiple processors writing to the same memory location, there is a race condition which reduces the parallelism making it a serial execution for preserving the correctness of the computation. But if we have extra space we might mitigate this problem of reduced parallelism by using a splitter as shown in the figure below. So, an interesting question would then be optimizing the parallelism with given some extra space. This is the main motivation for our problem, space time trade-off in parallel computations. This problem appeared in disguise to us in different areas of optimization like, Operational Research, Project Planning, Production Scheduling etc. Detailed work is mentioned in the related work section. 

{\bf[END OLD TEXT]}
}

%% file: related_work.tex

\subsection*{Related Work}
\label{sec:related}

While several prior works either directly or indirectly address Questions \ref{ques:nonreuse} (nonreusable resources) and \ref{ques:reuse-global} (globally reusable resources), to the best of our knowledge, Question \ref{ques:reuse-flow} (reusable along flow paths) has not been considered before.

The well-known time-cost tradeoff problem (TCTP) is closely related to our nonreusable resources question. In TCTP, some activities are expediated at additional cost so that the makespan can be shortened. Deadline and budget problems are two TCTP variants with different objectives. While the deadline problem seeks to minimize the total cost to satisfy a given deadline, the budget problem aims to minimize the project duration to meet the given budget constraint \cite{akkan2005network}. Most researchers consider the tradeoff functions to be either linear continuous or discrete giving rise to linear TCTP and discrete TCTP, respectively. 

Linear TCTP was formulated by Kelley and Walker in 1959 \cite{kelley1959critical}. They assumed affine linear and decreasing tradeoff functions. In 1961, linear TCTP was solved in polynomial time using network flow approaches independently by Fulkerson \cite{fulkerson1961network} and Kelley \cite{kelley1961critical}. Phillips and Dessouky \cite{phillips1977solving} later improved that result.

In 1997, De et al. \cite{de1997complexity} proved that discrete TCTP is NP-hard. 
For this problem, Skutella \cite{skutella1998approximation} proposed the first approximation algorithm under budget constraints which achieves an approximation ratio of $\Oh{\log{r}}$, where $r$ is the ratio of the maximum duration of any activity to the minimum one. Discrete TCTP can also be used to approximate the TCTP with general time-cost tradeoff functions, see, e.g., Panagiotakopoulos \cite{panagiotakopoulos1977cpm} and Robinson \cite{robinson1975dynamic}. For details on discrete TCTP see De et al. \cite{de1995discrete}.

Our problem with globally reusable resources (Question \ref{ques:reuse-global}) is very similar to the problem of scheduling precedence-constrained malleable tasks \cite{turek1992approximate}. In 1978, Lenstra and Rinnooy Kan \cite{lenstra1978complexity} showed that no polynomial time algorithm exists with approximation ratio less than $\frac{4}{3}$ unless $\textrm{P} = \textrm{NP}$. About 20 years later, Du and Leung \cite{du1989complexity} showed that the problem is strongly NP-hard even for two units of resources. In 2002, under the monotonous penalty assumptions of Blayo et al \cite{blayo1999dynamic}, Lep\`ere et al. \cite{lepere2002approximation} first proposed the idea of two-step algorithms -- computing an allocation first, and then scheduling tasks, and used this idea \cite{lepere2002approximationb} to design a algorithm that achieve an approximation ratio of $\approx 5.236$. In the first phase, they approximate an allocation using Skutella's algorithm \cite{skutella1998approximation}. Similarly, based on Skutella's approximation algorithm, Jansen and Zhang \cite{jansen2006approximation} devised a two-phase approximation algorithm with the best-known ratio of $\approx 4.730598$ and showed that the ratio is tight when the problem size is large. For more details on the problems of scheduling malleable tasks with precedence constraints, please check Dutot et al. \cite{dutot2004scheduling}.

There are memory allocators based on global memory manager for multi-core or multi-threaded systems such as scalloc \cite{aigner2015fast}, Hoard \cite{berger2000hoard}, llalloc \cite{llalloc}, Streamflow \cite{schneider2006scalable}, and TCMalloc \cite{TCMalloc}. They use thread-local space for memory allocation and a global manager for memory deallocation/reuse. For the global manager, they use concurrent data structures. However, these data structures can not completely avoid the need for synchronization \cite{aigner2015fast, henzinger2013quantitative, shavit2011data} without compromising correctness.

%% file: problem_formulation.tex
\newcommand\degreeSum[0]{degree\_sum }
\newcommand\makeSpan[0]{makespan }
\newcommand\kWayReducer[0]{\textit{$k$-way-split reducer }}
\newcommand\binaryReducer[0]{\textit{recursive binary split reducer }}
\newcommand\duration{duration}

\section{Preliminaries, Problem Formulation}
\label{sec:probs}

In general, the option to use reducers to trade off between extra space and the time to complete race-free writing operations leads to a {\em discrete resource-time tradeoff problem}, where, here, the valuable ``resource'' is the space that is added, in order to reduce the time necessary for the write operations. By investing in additional space, we can reduce the time it takes to do conflict-free write operations.

We formalize the discrete resource-time tradeoff problem. Consider a DAG, $D = (V,E)$, whose vertices $V$ correspond to jobs, and whose edges represent precedence relations among jobs. Without loss of generality, we assume that the DAG has a single source and a single sink vertex. The duration of a job depends on how much resource it receives. For each job $v\in V$, there is a non-increasing duration function $t_v(r)$ that denotes the time required to complete job $v$ using $r$ units of resources.
We call $\langle r,t_v(r)\rangle$ a {\em resource-time tuple} associated with job (vertex) $v$. 
We consider three classes of duration functions -- general non-increasing step functions, $k$-way splitting functions, and recursive binary splitting functions. 

\para{General non-increasing step function.}
Let $l_v$ be the number of resource-time tuples associated with job $v$. The $i$-th resource-time tuple is $\langle r_{v,i}, t_v(r_{v,i})\rangle$ where $1\le i \le l_v$. Then, the duration function $t_v(r)$ is a step function with $l_v$ steps described as follows:
 \begin{equation}
 \label{eq:duration-function-arb}
 t_v(r)=
 \begin{cases}
 t_v(r_{v,i}), & \text{if}\ r_{v,i}\le r < r_{v,i+1}, 1 \le i < l_v, \\
 t_v(r_{v,l_v}). & \text{if}\ r_{v,l_v} \le r,
 \end{cases}
 \end{equation} 
%
where $r_{v,1} = 0, r_{v,j} < r_{v,j + 1}$ and $t_v(r_{v,j}) \ge t_v(r_{v,j + 1})$ for $1 \leq j < l_v$.

\hide{
 \begin{equation}
 t_v(r)=
 \begin{cases}
\infty, & \text{if}\ 0\le r < r_{v,1}, \\
 t_v(r_{v,i}), & \text{if}\ r_{v,i}\le r < r_{v,i+1}, 1 \le i < l_v, \\
 t_v(r_{v,l_v}). & \text{if}\ r_{v,l_v} \le r,
 \end{cases}
 \end{equation} 
%
where $r_{v,j} < r_{v,j + 1}$ and $t_v(r_{v,j}) \ge t_v(r_{v,j + 1})$ for $1 \leq j < l_v$.
}
%
 
\para{$k$-way splitting.}
A {\em $k$-way split reducer} utilizes $k$ units of extra space, $S_v = \{s_1, s_{2},..,s_k\}$, associated with a vertex $v$, with $2 \le k \le d^{(in)}_{v}$, such that the write operations associated with incoming edges at $v$ are distributed among the vertices $S_{v}$, which then have edges linking each $s_i$ to $v$.
\hide{Refer to Figure~\ref{fig:k-way}, where a vertex $v$ with $m=d_v$ incoming edges becomes a vertex of in-degree $k$, by means of a $k$-way splitter, whose vertices $s_i$ each have in-degree $\lceil[\big]{m/k}\rceil $ incoming edges. Note that the in-degree contribution, $m=d_v$, of vertex $v$ to any path $\pi$ through $v$ is changed, via this $k$-way splitter at $v$, to $(\lceil[\big]{d_{v}/k}\rceil + k)$.}
The duration function that results from $k$-way split reducers is given by
 \begin{equation}
 t_v(r)=
 \begin{cases}
 t_v(0), & \text{if}\ k \in \{0,1\}  \\
 \lceil{t_v(0)/k}\rceil  + k, & \text{if}\ 2 \le k \le \lfloor\sqrt{t_v(0)}\rfloor \\
 t_v(\lfloor\sqrt{t_v(0)}\rfloor). & \text{if}\  \lfloor\sqrt{t_v(0)}\rfloor < k.
 \end{cases}
 \end{equation}

 \para{Recursive binary splitting.}
 The duration function that results from a recursive binary split reducer is given by a step function, as follows.
 The resource-time tuples are defined for $r = 0$ and $2^i$ where $0\le i\le k$   and $k = \lfloor \log_2 t_v(0) - \log_2 \log_2 e\rfloor$. The duration function $t_v(2^k) = \lceil{t_v(0)/2^k}\rceil  + k +1$ is minimized when $k= \lfloor \log_2 t_v(0) - \log_2 \log_2 e\rfloor$ (by differentiating $t_v(2^k)$ w.r.t. $k$). 
 
  \begin{equation}
  t_v(r)=
  \begin{cases}
  t_v(0), & \text{if}\ r = 0,1  \\
  \lceil{t_v(0)/2^i}\rceil  + i + 1 ,  & \text{if}\ r = 2^i, 2\le i\le k \\
  t_v(2^i) ,  & \text{if}\ 2^i \le r < 2^{i+1}, 2\le i\le k \\
  t_v(2^k) ,  & \text{if}\ i > k
  \end{cases}
  \end{equation}
  
When utilizing a reducer, extra space serves as the limited resource and the time taken for race-free writing at a vertex $v$ is the duration of the job corresponding to $v$. Both the $k$-way splitting duration function and the recursive binary splitting duration function are special cases of general non-increasing function. 


We consider jobs whose duration functions are of the types described above, and we distinguish between two optimization problems, depending on the objective function:
    
\para{Minimum-Makespan Problem.} Given a resource budget of $B$, assign the resources to the vertices $V$ such that the makespan of the project is minimized. Resources can be reused over a path.
  
\para{Minimum-Resource Problem.}
 Given a makespan target of $T$, minimize the amount of resources to achieve target makespan. Resources can be reused over a path. 
 

Finally, we remark that instead of jobs corresponding to vertices of the DAG, we can transform the DAG $D$ into another DAG $D'$ in which jobs correspond to edges of $D'$, and the precedence relations among jobs are enforced by introducing dummy edges, as follows: For each node $v$ in $D$, we introduce an edge $e_{v}=(a_v,b_v)$ in $D'$ (which then has the corresponding duration function, specified, e.g., by resource-time tuples). For each edge $(u,v)$ of $D$, we introduce a dummy edge, $e=(b_u,a_v)$ in $D'$, from the endpoint $b_u$ of edge $e_u=(a_u,b_u)$ to the origin $a_v$ of edge $e_v=(a_v,b_v)$, with resource-time function $t_{e}(r) = 0$ for all valid resource levels $r$. 


%% file: approx.tex
\section{Approximation Algorithms}
\label{sec:approx}


\subsection{Bi-criteria Approximation for Non-increasing Duration Functions}
\label{bi-criteria-general}

We use linear programming in our approximation algorithms. First, we relax the discrete duration function to a linear one. We transform the DAG so that a relaxed linear non-increasing duration function can be used. The transformation happens in two steps.

\para{Activity on arc reduction.}  
We reduce the input DAG $D$ into an equivalent DAG $D^{'}$ with activities on arcs instead of nodes. This is a simple transformation described earlier in Section \ref{sec:probs}.

\hide{
In the first step, we reduce the DAG $D$ into another DAG $D^{'}$ where jobs are represented by edges and the precedence relations among jobs are enforced by introducing dummy edges. For each node $v$ in $D$, we introduce an edge $e_{v}$ in $D^{'}$. The duration function of node $v$ from $D$ is assigned to edge $e_{v}$ of $D^{'}$. For each edge $(u,v)$ in $D$, we add a dummy edge in $D^{'}$ from the end point of edge $e_u$ to the starting point of edge $e_v$. The resource-time function for any dummy edge $e$ is $t_{e}(r) = 0$ for all valid $r$. 
}

\hide{Joe: here is the rephrasing of the above paragraph that I included earlier now:
  Finally, we remark that instead of jobs corresponding to vertices of the DAG, we can transform the DAG $D$ into another DAG $D'$ in which jobs correspond to edges of $D'$, and the precedence relations among jobs are enforced by introducing dummy edges, as follows: For each node $v$ in $D$, we introduce an edge $e_{v}=(a_v,b_v)$ in $D'$ (which then has the corresponding duration function, specified, e.g., by resource-time tuples). For each edge $(u,v)$ of $D$, we introduce a dummy edge, $e=(b_u,a_v)$ in $D'$, from the endpoint $b_u$ of edge $e_u=(a_u,b_u)$ to the origin $a_v$ of edge $e_v=(a_v,b_v)$, with resource-time function $t_{e}(r) = 0$ for all valid resource levels $r$. 
}

\begin{figure*}[h!]
\begin{minipage}{\textwidth}
\vspace{2cm}
	\centering
	\includegraphics[scale=0.4]{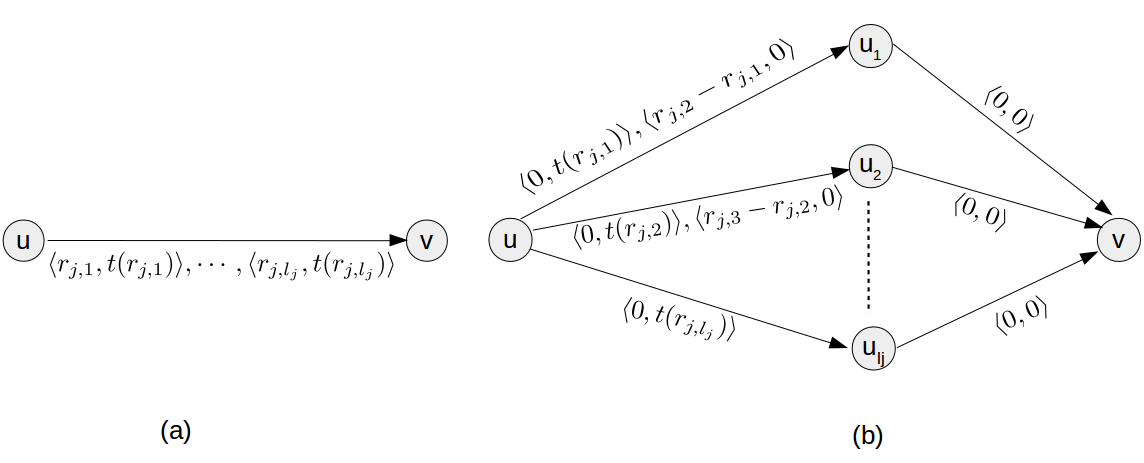}
	\vspace{-0.1cm}
	\figcaption{Transforming {\bf{\boldmath{$(a)$}}} a DAG with $l_j \ge 2$ resource-time tuples on each arc into {\bf{\boldmath{$(b)$}}} one with at most two resource-time tuples on each arc (Section \ref{bi-criteria-general})}.
	\label{fig:two-tuples}

	\vspace{0.5cm}

	\includegraphics[scale=0.4]{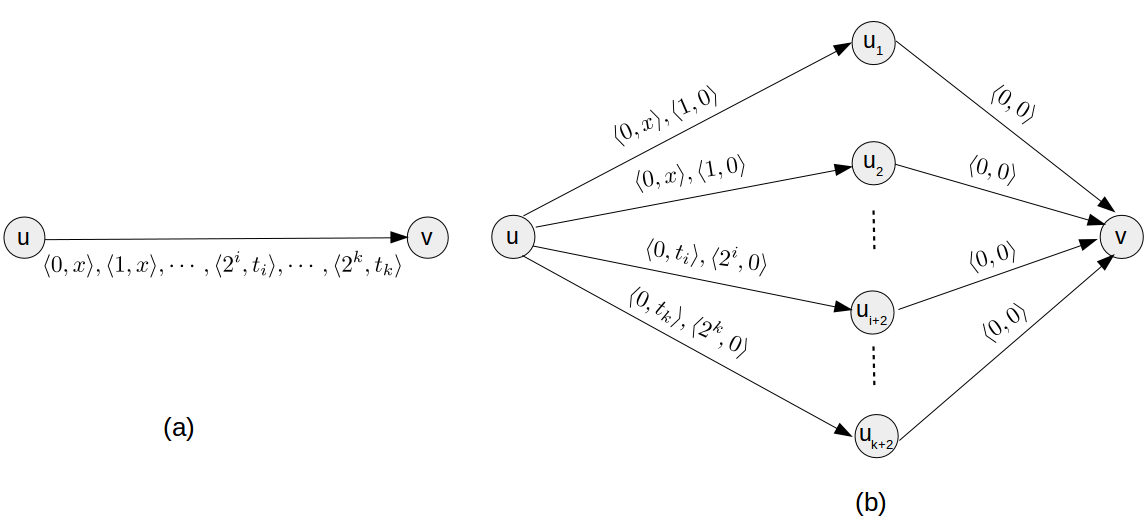}
	\vspace{-0.3cm}
	\figcaption{Transforming {\bf{\boldmath{$(a)$}}} a DAG with $(k+1)$ resource-time tuples on each arc based on the recursive binary splitting function into {\bf{\boldmath{$(b)$}}} one with at most two resource-time tuples on each arc (Section \ref{bi-criteria-bin})}.
	\label{fig:binary_approx}
\vspace{2cm}
\end{minipage}
\end{figure*}

\hide{
\begin{figure*}[h!]
	\centering
	\includegraphics[scale=0.35]{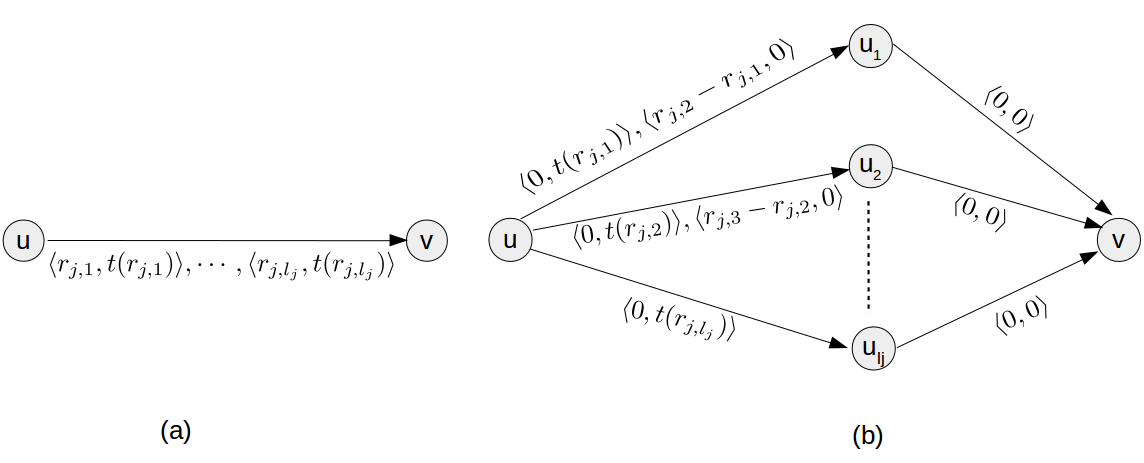}
	\vspace{-0.3cm}
	\caption{Transforming {\bf{\boldmath{$(a)$}}} a DAG with an general number of resource-time tuples on each arc into {\bf{\boldmath{$(b)$}}} one with at most two resource-time tuples on each arc (Section \ref{bi-criteria-general})}.
	\label{fig:two-tuples}
\end{figure*}
}

\para{Activity with two tuples.}
Following \cite{skutella1998approximation}, we create a DAG $D^{''}$ from $D^{'}$ such that all activities in $D^{''}$ are still on arcs and each such activity has at most $2$ resource-time tuples as shown in Figure \ref{fig:two-tuples}(b). 
%
%
Let $j$ be a job with $l_j \ge 2$ resource-time tuples $\langle r_{j, i}, t_j(r_{j, i})\rangle, 1 \leq i \leq l_{j}$ with $0 = r_{j,1} < r_{j,2} <\cdots < r_{j,l_{j}}$ and $t_j(r_{j,1}) \ge t_j(r_{j,2}) \ge \cdots \ge t_j(r_{j,l_{j}})$ (following Equation \ref{eq:duration-function-arb}). Let edge $(u,v)$ of $D^{'}$ represent job $j$. We add $l_j$ parallel chains, each consisting of two edges in $D^{''}$ (Figure \ref{fig:two-tuples}). For $1\le i\le l_j$, we create a chain of two edges $(u, u_i)$ and $(u_i, v)$.
We create a job $j_i$ for arc $(u, u_i)$ and associate two resource-time tuples with it.  
For $1 \le i < l_j$, job $j_i$ can be finished either using $0$ resource in $t_j(r_{j,i})$ units of time or using $(r_{j,i+1} - r_{j,i})$ units of resource in $0$ unit of time. The logic is that job $j$'s duration can be reduced from $t_j(r_{j,i})$ to $t_j(r_{j,i+1})$ provided the resource difference $(r_{j,i+1} - r_{j,i})$ is allocated to $j_i$. Thus the duration function is $t_{j_i}(0) = t_j(r_{j,i})$ and $t_{j_i}(r_{j,i+1} - r_{j,i}) = 0$. Job $j_{l_{j}}$'s (bottom most edge in the $l_j$ parallel edges for job $j$) duration cannot be further improved from $t_j(r_{j,l_j})$ units of time by using extra resources. The resource-time tuple at edge $(u_i,v)$ is $\langle0,0\rangle$ where $1\le i\le l_j$. 

There is a canonical mapping of resource usages and durations for jobs $j_i$ to that of job $j$. Let $x_i$ be the units of resource used for job $j_i$, then for job $j$, $\sum_{i=1}^{l_j}x_i$ units of resource are used. The time taken to finish job $j$ is $\max\{t_{j_i}(x_i) | 1 \le i \le l_j\}$. Without loss of generality, if we use $0$ unit of resource for job $j_i$ if $t_{j,i}(0) \le \max\{t_{j,1}(x_1), t_{j,2}(x_2), \cdots , t_{j,i-1}(x_{i-1})\}$, then this mapping is bijective. Thus we get the following lemma.  

\begin{lemma}
	Any approximation algorithm $\mathcal{A}$ on DAG $D^{''}$ (activity on edge and each edge has at most two resource-time tuples) with an approximation ratio $\alpha$ implies an approximation algorithm with the same approximation ratio $\alpha$ on general DAG $D$ (activity on vertex and each job can have more than two resource-time tuples). 	
\end{lemma}

From now on, we will only consider DAGs whose edges represent jobs, with each edge having at most two resource-time tuples.

\para{Linear relaxation.}
In $D^{''}$, any edge $(u,v)$ can have either two resource-time tuples $\{\langle 0, t_{(u,v)}(0) \rangle, \langle r_{(u,v)}, 0 \rangle \}$ or a single resource-time tuple $\{\langle 0, t_{(u,v)}(0) \rangle\}$. 
With linear relaxation, $r \in [0, r_{(u,v)}]$ units of resource can be used to reduce the completion time of the job corresponding to edge $(u,v)$ that has two resource-time tuples. The corresponding duration function $t_{(u,v)}(r)$ is as follows:

\begin{equation}
t_{(u,v)}(r) = \frac{t_{(u,v)}(0)}{r_{(u,v)}}r \text{   for  } r \in [0,r_{(u,v)}]
\end{equation}

The linear duration function $t_{(u,v)}(r)$ for the job $(u,v)$ with single resource-time tuple is as follows:

\begin{equation}
t_{(u,v)}(r) = t_{(u,v)}(0) \text{  for all}\ r \ge 0
\end{equation}

\hide{
With linear relaxation, $r \in [0, r_{(u,v)}]$ units of resource can be used to reduce the completion time of the job corresponding to edge $(u,v)$. The non-increasing function $t_{(u,v)}(r)$ that denotes the time taken to finish the job $(u,v)$ using $r \in [0, r_{(u,v)}]$ units of resource is as follows.
\hide{
\begin{equation}
t_{(u,v)}(r)=
\begin{cases}
\frac{t_{(u,v)}(0)}{r_{(u,v)}}r, & \begin{tabular}{@{}c@{}} if $0\le r < r_{(u,v)}$ and\\ $(u,v)$ has two tuples\end{tabular} \\
0, & \begin{tabular}{@{}c@{}} if $r_{(u,v)} \le r$ and\\ $(u,v)$ has two  tuples\end{tabular} \\
t_{(u,v)}(0), & \begin{tabular}{@{}c@{}} for any valid $r$ and\\ if $(u,v)$ has one tuple\end{tabular}
\end{cases}
\end{equation}  
}
%
\begin{equation*}
t_{(u,v)}(r)=
\begin{cases}
\frac{t_{(u,v)}(0)}{r_{(u,v)}}r, & \text{if}\ 0\le r < r_{(u,v)} \text{ and }\ (u,v) \text{ has two  tuples}\ \\
0, & \text{if}\ r_{(u,v)} \le r \text{ and }\ (u,v) \text{ has two  tuples}\ \\
t_{(u,v)}(0), & \text{for any valid }\ r \text{ and if }\ (u,v) \text{ with one tuple}\
\end{cases}
\end{equation*}  
}

\para{Linear programming formulation.}
Since we are allowed to reuse resources over a path we can model the problem as a network flow problem where resources are allowed to flow from the source to the sink in the DAG $D^{''}$. Let $E$ be the set of edges in $D^{''}$.
%
%
Let $f_{(u,v)}$ denote the amount of resources that flow through the edge $(u,v)$. Using linear relaxation on edge $(u,v)$, the time taken to finish the activity is $t_{(u,v)}(f_{(u,v)})$. Let the vertices in $D^{''}$ denote events. From now onwards, we use a vertex and its corresponding event synonymously. Let $E_v = \{(x,v)\}$ be the set of edges that are incident on vertex $v$. Event $v$ occurs if and only if all the jobs corresponding to the edges in set $E_v$ are finished. Let $T_v$ denote the time when event $v$ occurs. Let $s$ and $t$ denote the source vertex and the sink vertex, respectively. For source vertex $s$, we assume $T_s = 0$. All variables are non-negative.

\textbf{Constraints:}
\begin{equation}
\label{two-tuples-upper}
f_{(u,v)} \le r_{(u,v)} \text{  , }\ \forall (u,v)  \text{   with two resource-time tuples.}
\end{equation}

\begin{equation}
 \label{order}
T_{u} + t_{u,v}(f_{(u,v)}) \le T_{v}   \text{  , }\ \forall (u,v)\in E
\end{equation}
 
 \begin{equation}
 \label{conservation}
 \sum_{w}f_{(v,w)} + \sum_{u}f_{(u,v)} = 0 \text{  , }\ \forall v\notin \{s,t\}
 \end{equation}
 
\begin{equation}
\label{budget}
\sum_{k}f_{(s,k)} \le B
\end{equation}

\textbf{Objective function:}
\begin{equation}
\label{first-lp}
\min T_{t} 
\end{equation}

Inequality \ref{two-tuples-upper} upper bounds the resource flow variable $f_{(u,v)}$ for edges with two tuples. This ensures that these variables remain in the range $[0,r_{(u,v)}]$ and the duration function is linear in this range. Note that there is no such upper bound on the edges with single resource-time tuple (except the trivial total resource budget $B$ upper bound). This allows the flow of more resources over an edge that can be used later on a path.
Equation \ref{conservation} is a flow conservation constraint for all the vertices $v\notin \{s,t\}$. Inequality \ref{budget} constrains the flow of resources from source $s$ to be upper bounded by the resource budget.

\para{Solving the LP and rounding.}
We first solve the LP described above. This might give solution as fractional flow $f^{*}_{e}$  and duration $t_{e}(f^{*}_{e})$ at edge $e = (u,v)$. Let the resource-time tuples at edge $e$ be $\{\langle0, t_e(0)\rangle, \langle r_e, 0\rangle\}$. The range of feasible duration of activity $e$ is $[0, t_{e}(0)]$. We divide this range into two parts $[0, \alpha t_{e}(0)), [\alpha t_{e}(0), t_{e}(0)]$ where $0 < \alpha < 1$. If $t_{e}(f^{*}_{e}) \in [0, \alpha t_{e}(0) )$ we round it down to $0$, otherwise, we round it up to $t_{e}(0)$. Observe that in the first case, the resource requirement at $e$ can be increased by at most a factor of $1/(1-\alpha)$. In the second case, the completion time can be increased at most by a factor of $1/\alpha$. Let $f^{'}_e$ denote the rounded integer resource requirement at edge $e$.

\para{Computing min-flow.}
After rounding the LP solution, we get an integral resource requirement $f^{'}_e \in \{0, r_e\}$ for every edge $e$. We now compute a min-flow through this DAG where $f^{'}_e$ serves as the lower bound on the flow through (or resource requirement at) edge $e$. 

\hide{
Let $f$ be the min-flow solution. We prove that $f$ is a $1/(1-\alpha)$ approximation of the optimal flow $f^{*}$ (from the LP solution).
}

\textbf{Constraints:}
\begin{equation}
 \label{round-lp-c1}
f_{(u,v)} \ge f^{'}_{(u,v)}  \text{  , }\ \forall (u,v)\in E 
\end{equation}

\begin{equation}
 \label{round-lp-c2}
\sum_{w}f_{(v,w)} + \sum_{u}f_{(u,v)} = 0 \text{  , }\ \forall v\notin \{s,t\}
\end{equation}

\textbf{Objective function:}
\begin{equation}
\label{round-lp}
\min\sum_{k}f_{(s,k)}
\end{equation}

Let, $f$ and $f^*$ be the optimal solutions of LP \ref{round-lp-c1}--\ref{round-lp} and LP \ref{two-tuples-upper}--\ref{first-lp}, respectively.   

\begin{lemma}
	\label{feasible}
 $f^*/(1-\alpha)$ is a feasible solution of min-flow LP \ref{round-lp-c1}--\ref{round-lp}.
\end{lemma}
\begin{proof}
	Let $f^{*}_{e}$ be the optimal solution of LP \ref{two-tuples-upper}--\ref{first-lp}. We know that $f^{'}_e \le f^*_e/(1-\alpha)$. Hence, $f^*/(1-\alpha)$ is a feasible solution of that LP as it meets the resource requirement $f^{'}_e$ at every edge $e$.
\end{proof}

\begin{lemma}
	\label{feasible-alpha}
  $f$ is an integral flow and $f \le f^*/(1-\alpha)$, where $0 < \alpha < 1$.	
\end{lemma}

\begin{proof}
	The minflow problem has integral optimality. If $f$ is the optimal solution then it is an integral flow. From lemma \ref{feasible} we know that $f^*/(1-\alpha)$ is a feasible solution of LP \ref{round-lp-c1}--\ref{round-lp}. Since $f$ is optimal and $f^*/(1-\alpha)$ is a feasible flow, we have, $f \le f^*/(1-\alpha)$. 
\end{proof}

\para{Bi-criteria approximation.} We now summarize our bi-criteria approximation result for general non-increasing duration functions:

\begin{theorem}\label{thm:bi-criteria-arb}
	For any $\alpha \in (0, 1)$, there is a $(1/\alpha, 1/(1-\alpha))$ bi-criteria approximation algorithm for the discrete resource-time tradeoff problem with an general non-increasing duration function which allows resource reuse over paths. 
\end{theorem}

\begin{proof}
	First, we know from lemma \ref{feasible-alpha} that $f$ is an integral flow and $f \le f^*/(1-\alpha)$, where $0 < \alpha < 1$.	
	
	Second, we claim that the makespan of the DAG used in the minflow LP \ref{round-lp-c1}--\ref{round-lp} is at most a factor of $1/\alpha$ away from that of the LP \ref{two-tuples-upper}--\ref{first-lp} solution. Let us consider any $s-t$ path $\m{P}$. The makespan is at least the sum of completion times of the edges in $\m{P}$. Now, after rounding the LP \ref{two-tuples-upper}--\ref{first-lp} solution, the completion time of an edge may increase at most by a factor of $\alpha$. Hence, the sum of duration of edges along any path is increased at most by a factor of $\alpha$, thus the makespan will be increased by at most a factor of $\alpha$. 
\end{proof} 

\subsection{Single-criteria Approximation for $k$-Way and Recursive Binary Splitting}


First, observe the prior section gives us a bi-criterian approximation for both $k$-way and recursive binary splitting. Setting $\alpha = 1/2$ in Theorem~\ref{thm:bi-criteria-arb}, we obtain a $(2,2)$ bi-criteria approximation.%
\hide{
, in general. If  job $j$ has $l_j$ resource-time tuples, then DAG $D''$ has $l_j$ parallel jobs $j_i$ $(1 \le i \le l_j)$. After LP-rounding, each job $j_i$'s resource usage or duration can be at most a factor of $2$ away from that of LP solution. Recall that, job $j$'s resource usage is the sum of the resource assigned to all $j_i$ $(1 \le i \le l_j)$. Thus after rounding, the resource usage of job $j$ is at most a factor of $2$ away from that of LP solution. Similarly, the duration of job $j$ is the maximum duration among all $j_i$ $(1 \le i \le l_j)$. Thus after rounding, the duration of job $j$ is at most a factor of $2$ away from that of LP solution.
}
%
Now, after LP rounding, say a job $j$ uses $\overline{r_j}$ units of resource and takes $\overline{t_j}$ units of time.%
\hide{
Suppose the LP solution uses $r_{j}^{*}$ units of resource and takes $t_{j}^{*}$ units of time for job $j$. Please note that $\langle \overline{r_j}, \overline{t_j}\rangle$ is a valid resource-time tuple. From the LP solution,
}
%
Then the optimal solution uses $r_{j}^{*} \ge \overline{r_j}/2$ units of resource and takes $t_{j}^{*} \ge \overline{t_j}/2$ units of time for job $j$.
Recall that job $j$ consists of $l_j$ parallel jobs $j_i$ where $1\le i \le l_j$. Hence,  $\overline{r_j}$ is the sum of the resource (after rounding) used by  $l_j$ parallel jobs and $\overline{t_j}$ is the maximum time (after rounding) taken by $l_j$ parallel jobs.

\para{Approximation algorithm for $k$-way splitting.}
To obtain a single-criteria approximation, in the case of $k$-way splitting, we use at most $r_{j}^{*}$ units of resource for job $j$. If $\overline{r_j} > r_{j}^{*}$, we reduce $\overline{r_j}$ to $k$ (a nonnegative integer) units of resource such that $k \le r_{j}^{*}$. Using $k$ units of resource, job $j$ takes $t_j(k)$ units of time to complete.

\begin{lemma}
	\label{lemma25}
	$ \lceil{d/k}\rceil + k \le 2.5\overline{t_j}$ for $\overline{r_j} > 3$ where $d = t_j(0)$ and $k = \lfloor \overline{r_j}/2\rfloor$.
\end{lemma}  
\begin{proof}
Since $k = \lfloor \overline{r_j}/2\rfloor \ge  \overline{r_j}/2.5$ for $\overline{r_j} > 3$, we have $\lceil{d/k}\rceil \le d/k + 1 \le 2.5d/\overline{r_j}+1 \le 2.5\lceil{d/\overline{r_j}}\rceil+1$.
Also since $k= \lfloor \overline{r_j}/2\rfloor \le \overline{r_j}+1$ and $2.5\overline{r_j} \ge \overline{r_j}+2$ for $\overline{r_j} > 3$, we have
$\lceil{d/k}\rceil+k \le 2.5\lceil{d/\overline{r_j}}\rceil+1+\overline{r_j}+1 \le 2.5\left(\lceil{d/\overline{r_j}}\rceil+\overline{r_j}\right)$.
Hence, $t_j(k) \le 2.5\overline{t_j}$.
\end{proof}

\hide{
	\begin{equation}
	k = \lfloor \overline{r_j}/2\rfloor \ge  \overline{r_j}/2.5 \text{   if  }\ \overline{r_j} > 3
	\end{equation}
	
	\begin{equation}
	\lceil{d/k}\rceil \le d/k + 1 \le 2.5d/\overline{r_j}+1 \le 2.5\lceil{d/\overline{r_j}}\rceil+1
	\end{equation}
	
	\begin{equation}
	k= \lfloor \overline{r_j}/2\rfloor \le \overline{r_j}+1
	\end{equation}
	Since, $2.5\overline{r_j} \ge \overline{r_j}+2$ for $\overline{r_j} > 3$
	\begin{equation}
	\lceil{d/k}\rceil+k \le 2.5\lceil{d/\overline{r_j}}\rceil+1+\overline{r_j}+1 \le 2.5\big(\lceil{d/\overline{r_j}}\rceil+\overline{r_j}\big)   
	\end{equation}
}

\begin{lemma}
	\label{approx5}
	If $\overline{r_j} > 3$ then $t_j(k) \le 5t_{j}^{*}$.
\end{lemma}
\begin{proof}
	We know $t_j(k) = \lceil{d/k}\rceil + k$ as $k \ge 4$.
	Also in lemma \ref{lemma25}, we prove $t_j(k) \le 2.5\overline{t_j}$. However, we show that $\overline{t_j} \le 2t_{j}^{*}$. Hence, combining these two results we get $t_j(k) \le 5t_{j}^{*}$. 
\end{proof}

\begin{lemma}
	\label{opt2units}
	If $t_{j}^{*} = d/4$ then $r_{j}^{*} \ge 2$. 
\end{lemma}
\begin{proof}
	Recall that in $D^{''}$, job $j$ is represented as $l_j$ parallel jobs $j_i$ where $1\le i \le l_j$. The resource-time tuples of jobs $j_1$ and $j_2$ are $\{\langle0, d \rangle, \langle 2, 0\rangle\}$ and $\{\langle0, \lceil{d/2}\rceil+2 \rangle, \langle 1, 0\rangle\}$, respectively. To attain $d/4$ duration, $j_1$ requires at least $3/2$ units of resource and job $j_2$ requires $1/2$ unit of resource (applying linear relaxation). Hence, $r_{j}^{*}\ge (3/2+1/2)=2$ units of resource to achieve $t_{j}^{*} = d/4$.  
\end{proof}

\begin{lemma}
	\label{approx4}
	If $\overline{r_j} \le 3$ then $t_j(k) \le 4t_{j}^{*}$.
\end{lemma}
\begin{proof}
If $\overline{r_j} \le 3$ and $r_{j}^{*} < 2$, then we round down $\overline{r_j}$ to $k = 0$.
So, from Lemma \ref{opt2units} it follows that after rounding down to $0$ unit of resource, job $j$ takes $d \le 4t^{*}$ units of time.
 
If $\overline{r_j} \le 3$ and $r_{j}^{*} \ge 2$, then we round $\overline{r_j}$ to $k = 2$. It is true that $t_j(2) \le 2t_j(3)$ because $(\lceil{d/2}\rceil+2) \le 2(\lceil{d/3}\rceil+3)$. Also, $t_j(3)\le t_j(\overline{r_j})\le 2t_{j}^{*}$. Combining this two results we get $t_j(2) \le 4t^{*}$.  
\end{proof}

\hide{
We use $k \le r_{j}^{*}$ units of resource for job $j$ and $t_j(k) \le 5t_{j}^{*}$.
We compute a min-flow in $D^{'}$ where $k$ is the resource requirement for job $j$. Note that we are now working on DAG $D^{'}$ that does not have $l_j$ parallel chains for job $j$. Let $f$ be the min flow from source of $D^{'}$ such that all the resource requirements are met. The flow $f^*$ from LP solution is also a valid flow for the resource requirement $k$ for job $j$ as $k \le r_{j}^{*}$. We know that min-flow gives an optimal integral solution. Thus $f \le f^{*}$.
}

So, now we have the following result.

\begin{theorem}
	There is a $5$-approximation algorithm for the minimum-makespan problem with $k$-way splitting duration function.
\end{theorem}
\begin{proof}
	Combining Lemmas \ref{approx4} and \ref{approx5} we get $t_j(k) \le 5t_{j}^{*}$ for all valid $\overline{r_j}$. This proves that the makespan is at most 5 times the optimal solution. We now calculate the total amount of resource required to flow from the source of $D^{'}$. We compute a min-flow in $D^{'}$ where $k$ is the resource requirement for job $j$. Note that we are now working on $D^{'}$ that does not have $l_j$ parallel chains for job $j$. Let $f$ be the min flow from the source of $D^{'}$ such that all the resource requirements are met. The flow $f^*$ from the LP solution before rounding is also a valid flow for the resource requirement $k$ for job $j$ as $k \le r_{j}^{*}$. We know that min-flow gives an optimal integral solution. Hence, $f \le f^{*}$. 
\end{proof}

\para{Approximation algorithm for recursive binary splitting.}
We have the following result.

\begin{theorem}
There is a $4$-approximation algorithm for the minimum-makespan problem with recursive binary splitting function.
\end{theorem}
\begin{proof}
\hide{
Setting $\alpha = 1/2$, we get a $(2,2)$ bi-criteria approximation. Let after rounding the LP solution, job $j$ uses $\overline{r_j}$ units of resource and takes $\overline{t_j}$ units of time. Let LP solution uses $r_{j}^{*}$ units of resource and takes $t_{j}^{*}$ units of time for job $j$. Please note that $\langle \overline{r_j}, \overline{t_j}\rangle$ is a valid resource-time tuple for job $j$. This also implies that $r_j$ is either $0$ or a power of $2$. We know, $r_{j}^{*} \ge \overline{r_j}/2$ and $t_j(r_{j}^{*}) = t_{j}^{*} \ge \overline{t_j}/2$.\\
}
%
As in the case of k-way splitter, to get a single-criteria approximation, we use no more than $r_{j}^{*}$ units of resource for job $j$. If $\overline{r_j} > r_{j}^{*}$, we reduce $\overline{r_j}$ to $\overline{r_j}/2$. We know that $t_j(\overline{r_j}/2) \le  2t_j(\overline{r_j})$ from the properties of the recursive binary splitting function. Thus, $t_j(\overline{r_j}/2) \le  2t_j(\overline{r_j}) \le 4t_j(r_{j}^{*}) = 4t_{j}^{*}$.
\end{proof}

%% file: binary_splitting_approx.tex

\subsection{Improved Bi-criteria Approximation for Recursive Binary Splitting Functions}
\label{bi-criteria-bin}

Putting $\alpha = 3/4$ in Theorem \ref{thm:bi-criteria-arb} we obtain a $(4/3, 4)$ bi-criteria approximation algorithm for general non-increasing duration functions. Hence, if we use $4/3$ times more resources than OPT (i.e., the optimal solution), we are guaranteed to get a makespan within factor of 4 of OPT. In this section we show that the bound can be improved to $(4/3, 14/5)$ for recursive binary splitting functions.

For a node with in-degree $x$, the resource-time tuples based on the recursive binary splitting function are as follows: $\{\langle 0,x\rangle, \langle 1,x\rangle, \langle 2, t_1\rangle,$ $... , \langle 2^i, t_i\rangle, \langle 2^{i+1}, t_{i+1}\rangle\,...,\langle 2^k, t_k\rangle\}$ where $t_j = \lceil{x/2^j}\rceil + j+1$ for $j \ge 2$ and $k= \lfloor \log_2 x - \log_2 \log_2 e\rfloor$ is the largest value of $j$ for which $t_j$ decreases with the increase of $j$. See Figure \ref{fig:binary_approx}.

\hide{
\begin{figure*}[h!]
	\centering
	\includegraphics[scale=0.30]{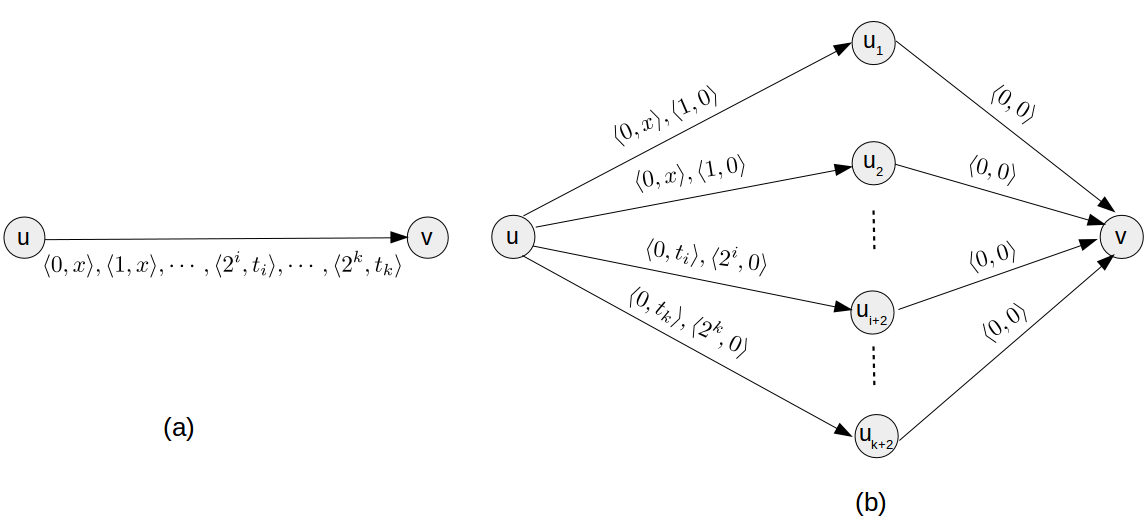}
	\vspace{-0.3cm}
	\caption{Transforming {\bf{\boldmath{$(a)$}}} a DAG with $(k+1)$ resource-time tuples on each arc based on the recursive binary splitting function into {\bf{\boldmath{$(b)$}}} one with at most two resource-time tuples on each arc (Section \ref{bi-criteria-bin})}.
	\label{fig:binary_approx}
\end{figure*}
}

After solving LP \ref{two-tuples-upper}--\ref{first-lp} from Section \ref{bi-criteria-general}, we sum up the (possibly fractional) resources allocated to all the $l_j$ parallel edges corresponding to job $j$. Let $r$ be that sum. Let $t$ be the maximum among the time values given by the LP solution for the $l_j$ parallel edges. Thus, the LP takes $t$ units of time for job $j$.

We round $r$ to an integer $\overline{r}$ based on the following criteria.

\begin{equation*}
\overline{r} = 
\begin{cases}
0, & \text{if}\ r < 1 \\
2^i & \text{if}\ 2^i \le r < (2^i + 2^{i+1})/2,  0\le i \le k \\
2^{i+1}, & \text{if}\ (2^i + 2^{i+1})/2 \le r < 2^{i+1},  0\le i \le k 
\end{cases}
\end{equation*}

We want to find a constant $\rho$, such that if $t = t_i/\rho$, then the LP must use at least  $(2^i + 2^{i+1})/2 = 3(2^{i-1})$ units of resources.
We compute $r$ as follows. 
In Figure \ref{fig:binary_approx}$(b)$, each of the top two edges $(u, u_1)$ and $(u,u_2)$ requires $(1-(1/x)t)$ units of resource to finish in time $t$. Each edge $(u, u_{j+2})$ for $1 \le j \le i+1$ requires $\Big(2^j - (2^j/t_j)t\Big)$ units of resource to finish in time $t$. Summing over all these edges, we get the expression of $r$
\begin{equation*}
\label{equate}
\begin{aligned}
r &=  2\left(1 - \frac{1}{x}t\right) + \sum_{j=1}^{i+1}\left(2^j-\frac{2^j}{t_j}t\right) 
  = 8\cdot (2^{i-1}) - \frac{t_i}{\rho}\left(2/x + \sum_{j=1}^{i+1}\frac{2^j}{t_j}\right)
\end{aligned}
\end{equation*}

\hide{
\begin{equation}
\label{equate}
\begin{aligned}
r &= \big(1 - \frac{1}{x}t\big) + \big(1 - \frac{1}{x}t\big) + \sum_{j=1}^{i+1}\big(2^j-\frac{2^j}{t_j}t\big) \\
&= \big(2 - \frac{2}{x}t\big) + \sum_{j=1}^{i+1}\big(2^j-\frac{2^j}{t_j}t\big) \\
&= \big(2+\sum_{j=1}^{i+1}2^j\big) - t\big(2/x + \sum_{j=1}^{i+1}\frac{2^j}{t_j}\big) \\
&= 1+ \frac{2^{i+2}-1}{2-1} - t\big(2/x + \sum_{j=1}^{i+1}\frac{2^j}{t_j}\big) \\
&= 8\cdot (2^{i-1}) - \frac{t_i}{\rho}\big(2/x + \sum_{j=1}^{i+1}\frac{2^j}{t_j}\big). 
\end{aligned}
\end{equation}
}

Since we want to have $r \ge 3(2^{i-1})$, we want to find the smallest value of $\rho$ such that

\begin{equation*}
\label{rho}
\begin{aligned}
 & \frac{t_i}{\rho}\left(2/x + \sum_{j=1}^{i+1}\frac{2^j}{t_j}\right) \le 5\cdot (2^{i-1}) \Rightarrow & \rho \ge 1/5 \left(\frac{t_i}{2^{i-2}x}+\sum_{j=1}^{i+1}\frac{t_i}{2^{i-j-1}t_j}\right). 
\end{aligned}
\end{equation*}

Now,

\begin{equation*}
\begin{aligned}
~~~~~~ & \frac{t_i}{2^{i-2}x}+\sum_{j=1}^{i+1}\frac{t_i}{2^{i-j-1}t_j}
= \frac{\lceil{\frac{x}{2^i}}\rceil+ i+1}{x(2^{i-2})} + \sum_{j=1}^{i+1}\frac{\lceil{\frac{x}{2^i}}\rceil+ i+1}{(\lceil{\frac{x}{2^j}}\rceil+ j+1)2^{i-j-1}} \\
&< \frac{\frac{x}{2^i} + i + 2}{x(2^{i-2})} + \sum_{j=1}^{i+1}\frac{\frac{x}{2^i} + i + 2}{(\frac{x}{2^j}+ j+1)2^{i-j-1}} \\
&=\frac{1}{2^i}\frac{1}{2^{i-2}} + \frac{i+2}{x(2^{i-2})} + \sum_{j=1}^{i+1}\frac{\frac{1}{2^{i-j}}(\frac{x}{2^j} + j+1) + i+2 - \frac{j}{2^{i-j}}-\frac{1}{2^{i-j}}}{(\frac{x}{2^j} + j+1)2^{i-j-1}} \\
&\le \left(\frac{i+2}{x}\frac{1}{2^{i-2}}+ \sum_{j=1}^{i+1}\frac{i+2}{(\frac{x}{2^j}+j+1)2^{i-j-1}}\right) + \left(\frac{1}{2^i}\frac{1}{2^{i-2}}+\sum_{j=1}^{i+1}\frac{1}{2^{i-j}}\frac{1}{2^{i-j-1}}\right) \\
&=\left(\frac{i+2}{x}\frac{1}{2^{i-2}}\right)+ \left(\sum_{j=1}^{i+1}\frac{i+2}{(\frac{x}{2^j}+j+1)2^{i-j-1}}\right) + \left(\frac{32}{3} + \frac{1}{3}\frac{1}{4^{i-1}}\right)
\end{aligned}
\end{equation*}

Let, $A = \frac{i+2}{x}\frac{1}{2^{i-2}}$ , $B = \sum_{j=1}^{i+1}\frac{i+2}{(\frac{x}{2^j}+j+1)2^{i-j-1}}$ and 
$C = 32/3 + \frac{1}{3}\frac{1}{4^{i-1}}$. 

Note that $i+2 = (i+1) + 1\le (\log_2 x-\log_2\log_2 e) +1$, since $i+1 \le k$.
Hence,
\begin{equation*}
\begin{aligned}
A &\le \frac{(\log_2 x-\log_2\log_2 e) +1}{x}\frac{1}{2^{i-2}}
\le \frac{2}{e}\frac{1}{2^{i-2}}.
\end{aligned}
\end{equation*} 
\\
Now, $x/2^j + j+1 \ge (\log_2 x - \log_2\log_2 e + \frac{1}{\ln{2}})$ and hence, 
\hide{
This is true because $\frac{2(\log_2 x-\log_2\log_2 e) +1}{x}$ is maximized when $x = \frac{e\log_2 e}{2}$.\\
}

\begin{equation*}
\begin{aligned}
B &\le \sum_{j=1}^{i+1}\frac{(\log_2 x -\log_2\log_2 e)+1}{(\log_2 x - \log_2\log_2 e + \frac{1}{\ln{2}}+1) }\frac{1}{2^{i-j-1}}\\ 
  & < \sum_{j=1}^{i+1}\frac{1}{2^{i-j-1}} 
= 2- \frac{1}{2^{i-2}}.
\end{aligned}
\end{equation*}
\\Thus, $A+B+C < \frac{2}{e}\frac{1}{2^{i-2}} + 2- \frac{1}{2^{i-2}} + 32/3 + \frac{1}{3}\frac{1}{4^{i-1}} \le 14$.
\\
Therefore, $(t_i/x)\frac{1}{2^{i-2}} + \sum_{j=1}^{i+1}\frac{t_i}{t_j}\frac{1}{2^{i-j-1}} < 14$.
\\
So, by setting $\rho = 14/5$, we get $\rho > 1/5\left((t_i/x)\frac{1}{2^{i-2}} + \sum_{j=1}^{i+1}\frac{t_i}{t_j}\frac{1}{2^{i-j-1}}\right)$.
\\
Summarizing, we get the following lemmas from the computation above. 
\begin{lemma}
	\label{binary_time}
	To achieve a duration of $t = t_i/(14/5)$ for any job $j$, the LP solution uses at least $3(2^{i-1})$ units of resources for $0\le i \le k$.
\end{lemma}

Lemma \ref{binary_time} implies the following.
\begin{lemma}
	\label{binary_time_ratio}
If the LP uses $2^i \le r < 3(2^{i-1})$ units of resources and we round $r$ down to $\overline{r} = 2^i$ where $0\le i \le k$, then $t_i \le (14/5)t$ where $t$ is the duration from the LP solution.
\end{lemma}

\begin{lemma}
	\label{binary_unit_resource}
	With $r < 1$ units of resource, the LP cannot achieve a duration of $t < x/2$ for job $j$. 
\end{lemma}
\begin{proof}
	The first edge has resource-time tuples $\{\langle0,x\rangle, \langle1,0\rangle\}$. To achieve a duration of $x/2$, the LP has to use $1/2$ unit of resource on the first edge. The second edge also has the same resource-time tuples $\{\langle0,x\rangle, \langle1,0\rangle\}$, and it also takes $1/2$ unit of resource. Thus, the first two edges alone need $1$ unit of resource to achieve a duration of $x/2$ for all $l_j$ parallel edges of job $j$.   
\end{proof}

Lemma \ref{binary_unit_resource} implies the following.
\begin{lemma}
	\label{binary_round_zero}
	If the LP uses $r<1$ unit of resource and we round $r$ down to  $0$, then $t_i \le 2t$, where $t$ is the duration from the LP solution.
\end{lemma}

\begin{lemma}
	\label{binary_resource_ratio}
	If $r$ rounded to $\overline{r}$ then $\overline{r} \le (4/3)r$
\end{lemma}
\begin{proof}
When we use $\overline{r} = 2^{i+1}$ units of resource after rounding, the LP uses at least $ 3(2^{i-1}) \le r \le 2^{i+1}$ units. Thus, $\overline{r} \le (4/3)r$.
\end{proof}

From Lemma \ref{binary_time_ratio} and Lemma \ref{binary_resource_ratio}, we get the following theorem.
\begin{theorem}
	\label{binary_thm}
	There is a $(4/3, 14/5)$ bi-criteria approximation algorithm for the discrete resource-time tradeoff problem with resource reuse along paths when the recursive binary duration function is used.
\end{theorem}

%% file: exact_algorithm.tex
\subsection{Exact Algorithm for Series-Parallel Graphs}

We consider now the special case in which the underlying DAG $D$ is a series-parallel graph.
A series-parallel graph $G$ can be transformed into (and represented as) a rooted binary tree $T_G$ in polynomial time
by decomposing it into its atomic parts according to its series and parallel compositions (see, e.g., \cite{mohring1989computationally}).
%
In $T_G$, the leaves correspond to the vertices of $G$. Internal nodes of $T_G$ are labeled as ``$s$" or ``$p$" based on series or parallel composition. We associate each internal node $v$ of $T_G$  with the series-parallel graph $G_v$, induced by the leaves of the subtree rooted at $v$. 

Let $T(v, \lambda)$ denote the makespan of $G_v$ using $0 \le \lambda \le B $ units of resources where $B$ is the resource budget. We want to solve for $T(s, B)$, where $s$ is the root of $T_G$. This can be done using dynamic programming, solving for the leaves first, and then progressing upward to the root of $T_G$. We compute $T(v, \lambda)$ as follows which assumes that node $v$ corresponds to job $j$ if it is a leaf, otherwise it has two children $v_1$ and $v_2$.

\hide{
$T(v, \lambda) = t_j(\lambda)$. If $v$ is an internal node with left child $v_1$ and right child $v_2$, then we compute $T(v, \lambda)$ as follows. If $v$ is an internal node with label ``$s$", then $T(v, \lambda) = T(v_1, \lambda) + T(v_2, \lambda)$. If $v$ is an internal node with label ``$p$", then $T(v, \lambda) = min_{0\le i\le\lambda}[max(T(v_1, i) , T(v_2, \lambda - i))]$. 
}

\begin{equation*}
T(v, \lambda) = 
\begin{cases}
\begin{tabular}{p{1.6in}l} $t_j(\lambda)$  & if $v$ is a leaf\end{tabular} \\
\begin{tabular}{p{1.6in}l} $T(v_1, \lambda) + T(v_2, \lambda)$ & \begin{tabular}{@{}l@{}} if $v$ is an internal\\node with label ``$s$''\end{tabular}\end{tabular} \\
\begin{tabular}{p{1.6in}l} $min_{0\le i\le\lambda}\left\{max\left\{\begin{tabular}{@{}c@{}} $T(v_1, i)$,\\ $T(v_2, \lambda - i)$\end{tabular}\right\}\right\}$ &\begin{tabular}{@{}l@{}}if $v$ is an internal\\node with label ``$p$''\end{tabular}\end{tabular} \\
\end{cases}
\end{equation*}

There are $\Oh{m}$ nodes in $T_G$ if $G$ has $m$ edges. For each node $v$ we compute $T(v,\lambda)$ for $0 \le \lambda \le B $. Computing $T(v,\lambda)$ for any particular value of $\lambda$ takes $\Oh{\lambda}$ time, since, if the node is a ``$p$" node, then for $0\le i\le\lambda$ we need to look up values $T(v_1, i)$. Thus, for any internal node $v$, it takes $\sum_{\lambda = 0}^{B}{\Oh{\lambda}} = \Oh{B^2}$ time. As there are $\Oh{m}$ nodes in $T_G$, the (pseudo-polynomial) time complexity of the algorithm is $\Oh{mB^2}$. 

%% file: hardness.tex
\section{NP-Hardness}
\label{sec:NPHardness}
In this section we give a variety of NP-hardness and inapproximability results related to the discrete time-resource tradeoff problem in the offline setting (i.e., when the entire DAG is available offline). All problems consider the version where there is resource reuse over paths, but they vary the cost-function, graph structure, and minimization goal. Section~\ref{general-hard} gives several reductions from 1-in-3SAT. Theorem~\ref{thm:strong-hardness} gives a base reduction for the problem with general non-increasing duration function which will provide the ideas and structure for later more complex proofs. Theorems~\ref{thm:hardness-approx-makespan} and \ref{thm:hardness-approx-resources} adapt this proof to give constant factor inapproximability for the minimum-resource and minimum-makespan problems. Section~\ref{sec:splitting-hardness} adapts the NP-hardness proof to apply when the cost function is restricted to be the recursive binary splitting and the $k$-way splitting.

Section~\ref{sec:BoundedTreewidthHardness} considers the problem in bounded treewidth graphs. We show weak NP-hardness by a reduction from Partition.

\begin{figure*}[h!]
\begin{minipage}{\textwidth}
\vspace{1cm}
\begin{tabular}{cc}
	\includegraphics[scale=0.25]{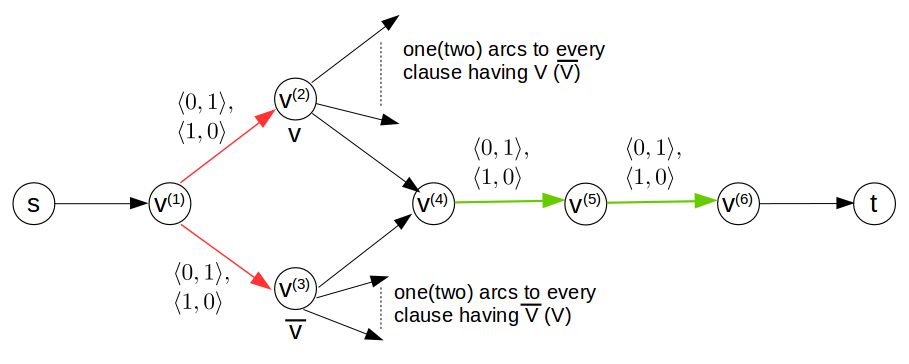} & \includegraphics[scale=0.23]{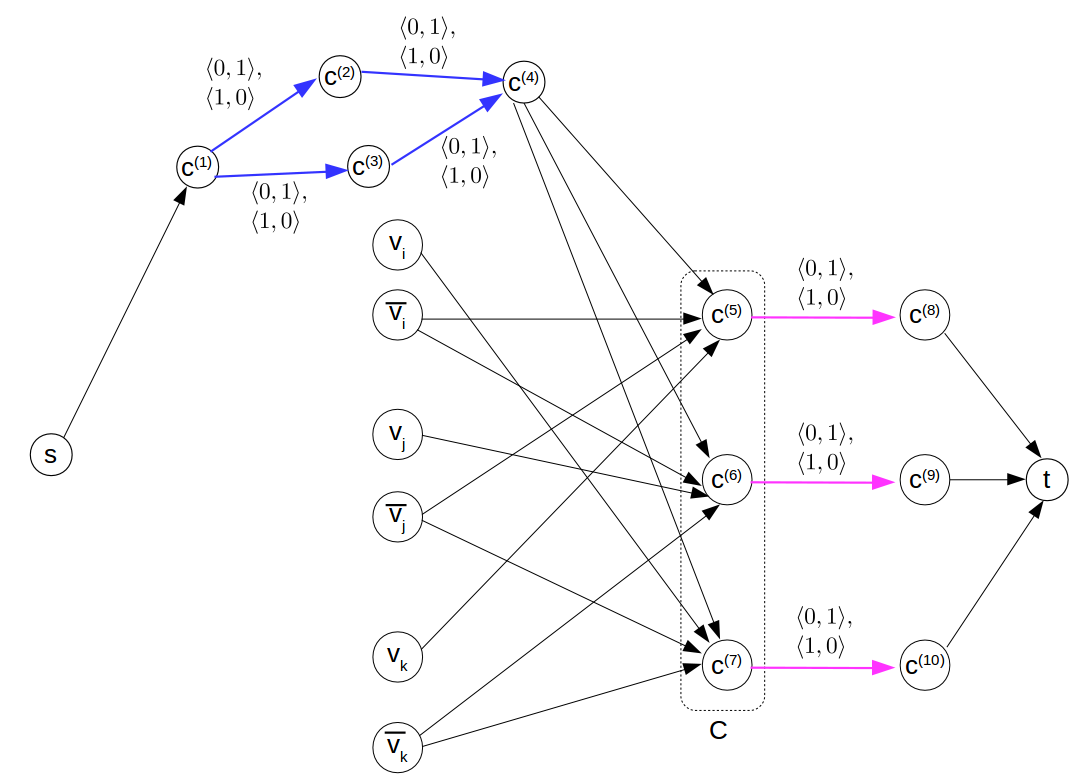}\\
	{\bf{\boldmath{$(a)$}}} & {\bf{\boldmath{$(b)$}}}\\
\end{tabular}
	\figcaption{{\bf{\boldmath{$(a)$}}} Gadget for variable V, and {\bf{\boldmath{$(b)$}}} gadget for clause $C = (V_i\lor V_j\lor V_k)$ (Section \ref{general-hard}).}
	\label{fig:var-clause}
	\vspace{0.5cm}

\hide{
\begin{figure*}[h!]
	\centering
	\includegraphics[scale=0.35]{pics/var.png}
	\caption{Gadget for variable V}
	\label{fig:var}
\end{figure*}

\begin{figure*}[h!]
	\centering
	\includegraphics[scale=0.32]{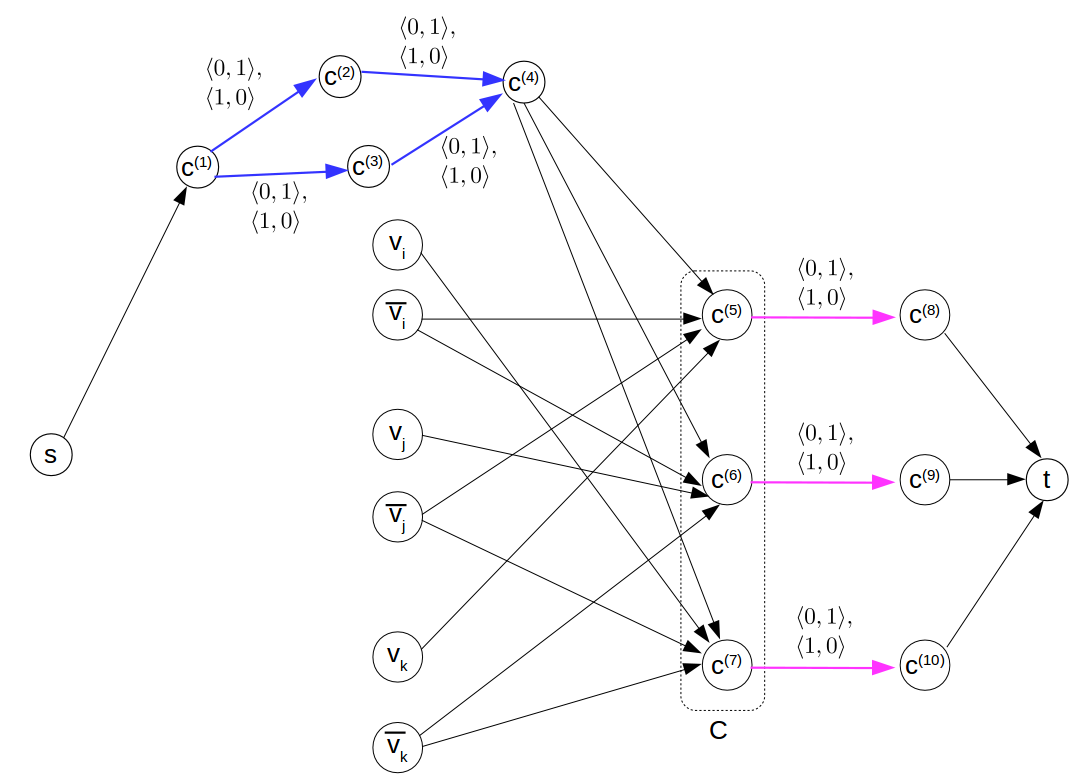}
	\caption{Clause gadget for $C = (V_i\lor V_j\lor V_k)$.}
	\label{fig:clause}
\end{figure*}
}

	\centering
	\includegraphics[scale=0.32]{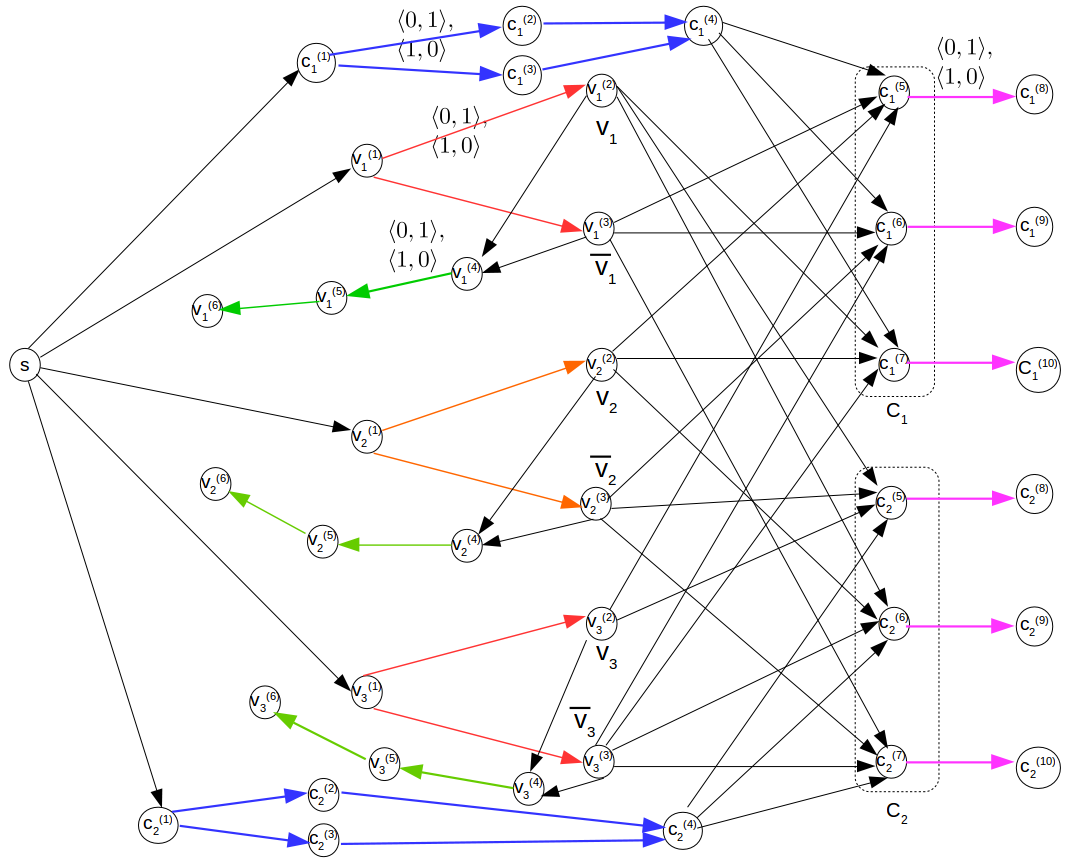}
	\figcaption{The complete construction for $(V_1\lor \neg V_2\lor V_3)\land (\neg V_1\lor V_2\lor V_3)$ is satisfiable with the truth assignment: $V_1 = $ TRUE, $V_2 = $ TRUE, $V_3 = $ FALSE (Section \ref{general-hard}).}
	\label{fig:example}
\vspace{1cm}
\end{minipage}
\end{figure*}

\vspace{-0.2cm}
\subsection{Reuse Over a Path with General Non-increasing Duration Function}
\label{general-hard}

\vspace{-0.1cm}
\begin{theorem}
  \label{thm:strong-hardness}
  It is (strongly) NP-hard to decide if there exists a solution to the (offline) discrete resource-time tradeoff problem, with resource reuse over paths and a non-increasing duration function, satisfying a resource bound $B$ and a makespan bound $T$.
\end{theorem}

Our proof is based on a polynomial-time reduction from the strongly NP-hard problem 1-in-3SAT \cite{schaefer1978complexity}:
Given $n$ variables ($V_i, 1\le i\le n$) and $m$ clauses ($C_j, 1\le j\le m$),with each clause a disjunction of three literals, is there a truth assignment to the variables such that each clause has exactly one true literal? 

\hide{
\begin{problem}
\textbf{1-in-3SAT :} Given $n$ variables ($V_i, 1\le i\le n$) and $m$ clauses ($C_j, 1\le j\le m$) where each clause is a disjunction of three literals, is there a truth assignment to the variables such that each clause has exactly one true literal? 
\end{problem}
}

\para{Variable gadget.} The gadget for variable $V$ consists of nodes $V^{(1)}$, $V^{(2)}$, $V^{(3)}$, $V^{(4)}$, $V^{(5)}$, and $V^{(6)}$ as shown in Figure \ref{fig:var-clause}$(a)$. We show in the hardness proof that a variable gadget will get exactly one unit of extra resource, otherwise the makespan will be greater than the target makespan of $1$. Sending one unit of resource to node $V^{(2)}$ (Figure \ref{fig:var-clause}$(a)$) corresponds to setting the variable $V$ to TRUE and sending the unit of resource to $V^{(3)}$ corresponds to setting $V$ to FALSE. The remaining vertices ensure the extra resource is used in the variable and not transferred into one of the clauses.

\para{Clause gadget.} The gadget corresponding to clause $C$ has $10$ vertices $C^{(i)}$ ($1\le i\le 10$) as shown in Figure \ref{fig:var-clause}$(b)$. Arcs $(C^{(1)}, C^{(2)})$, $(C^{(2)}, C^{(4)})$, $(C^{(1)}, C^{(3)})$ and $(C^{(3)}, C^{(4)})$ have resource-time pairs as $\{\langle 0, 1\rangle,\langle 1, 0\rangle\}$. If clause $C$ has three literals $V_i, V_j$ and $V_k$, then vertex $C^{(5)}$ is connected to the vertices $V_{i}^{(3)}, V_{j}^{(3)}$ and $V_{k}^{(2)}$. These vertices correspond to $\neg V_i, \neg V_j$ and $V_k$ respectively. Vertex $C^{(6)}$ is connected to $V_{i}^{(3)}, V_{j}^{(2)}$ and $V_{k}^{(3)}$ . These vertices correspond to $\neg V_i, V_j$ and $\neg V_k$. Vertex $C^{(7)}$ is connected to $V_{i}^{(2)}, V_{j}^{(3)}$ and $V_{k}^{(3)}$. These vertices correspond to $V_i, \neg V_j$ and $\neg V_k$. Arcs $(C^{(5)}, C^{(8)})$, $(C^{(6)}, C^{(9)})$, and $(C^{(7)}, C^{(10)})$ have resource-time pairs as $\{\langle 0, 1\rangle,\langle 1, 0\rangle\}$. The part of the clause gadget consisting of $C^{(1)}, C^{(2)}, C^{(3)}$ and $C^{(4)}$ demand at least two units of memory be allocated there and then these units of resource go to satisfy two of $C^{(5)},C^{(6)}$ and $C^{(7)}$. There is still one of these lines that has no allocated resource so it's cost is 1. Thus the corresponding variable must have had it's path length reduced (by setting it true).

Figure \ref{fig:example} shows the complete construction of $(V_1\lor \neg V_2\lor V_3)\land (\neg V_1\lor V_2\lor V_3)$ as an example.

\vspace{-0.1cm}
\begin{lemma}
	\label{lem:zero-one}
	There exists a solution to the input instance of 1-in-3SAT iff there exists a valid flow of resources through the DAG achieving a makespan of $1$ under a resource bound of $B = n + 2m$.
\end{lemma}

\vspace{-0.2cm}
\begin{proof}
\noindent{\bf Forward direction.} We prove that if there is a solution to the 1-in-3SAT instance with $n$ variables and $m$ clauses, then the reduced DAG has a solution of makespan $1$ with $(n+2m)$ units of resource.
If a variable $V$'s truth assignment is TRUE, then we allow one unit of resource to flow through vertex $V^{(2)}$ along the path $\langle S, V^{(1)}, V^{(2)}, V^{(4)}, V^{(5)}, V^{(6)}, T\rangle$, otherwise we allow one unit of resource to flow through vertex $V^{(3)}$ along the path $\langle S, V^{(1)}, V^{(3)}, V^{(4)}, V^{(5)}, V^{(6)}, T\rangle$. For every clause $C$, we allow one unit of resource to flow through the path $\langle S, C^{(1)}, C^{(2)}, C^{(4)}\rangle$ and another unit of resource through the path $\langle S, C^{(1)}, C^{(3)}, C^{(4)}\rangle$. Thus, $2$ units of resource can be flowed from vertex $C^{4}$.
In a valid assignment of 1-in-3SAT, for each clause $C$, exactly $2$ vertices of $C^{(5)}, C^{(6)}$ and $C^{(7)}$ will have the earliest start time of $1$ and the other one will have $0$ (Table \ref{table:1}). 

Also, if only one literal is true in a clause, exactly two vertices among $C^{(5)}, C^{(6)}$ and $C^{(7)}$ need one unit of extra resource each to meet the makespan requirement (from Table \ref{table:1}). We are allowed to flow $2$ units of resource from vertex $C^{(4)}$. Thus the project makespan is $1$ using $(n+2m)$ units of resource.

\noindent{\bf Backward direction.} Now, we prove that if there exists a solution of makespan $1$ using $(n+2m)$ units of resource in the reduced DAG, then there also exists a solution to the 1-in-3SAT instance. 
To achieve a makespan of $1$, every variable gadget needs $1$ unit of resource and each clause gadget needs $2$ units of resource, otherwise the makespan would be greater than $1$. Also, any resource that is used in a variable gadget cannot be used further in any other variable or clause gadget because the resource can be reused over a path only. Similarly, any resource that is used in any clause gadget, cannot be reused in any other gadget. Only one vertex that is either $V^{(2)}$ or $V^{(3)}$, will have the earliest start time $0$. Both cannot be $0$, as there is only $1$ unit of resource per variable gadget.
%
%
%
Both cannot be $1$ as in a clause $C$ where the literal $V$ or $\neg V$ is present, each of $C^{(5)}, C^{(6)}$ and $C^{(7)}$ would have earliest starting time of $1$. This requires use of $3$ units of resource in the clause gadget $C$ to achieve a makespan of $1$. However, each clause gadget can have exactly $2$ units of resource. Thus, for every variable, it has to be a valid assignment ($V$ is set to either TRUE or FALSE).
From Table \ref{table:1}, if a clause has exactly one TRUE literal, then the clause gadget requires $2$ units of resource to achieve a makespan of $1$. Otherwise, the clause gadget would have a makespan of $2$ with the same amount of resource or would require more resource to achieve the target makespan of $1$. Thus, each clause has exactly one TRUE literal. This satisfies the 1-in-3SAT instance.
\end{proof}
\vspace{-0.1cm}

\begin{table}[!ht]
	\centering
	\scalebox{0.88}[0.88]{
		{\small 
		\hspace{-0.25cm}	\begin{colortabular}{ | c c c c c c |}
				\hline                       
				\rowcolor{tabletitlecolor} $V_i$ & $V_j$ & $V_k$ & $C^{(5)}$ & $C^{(6)}$ & $C^{(7)}$ \\  \hline

			True & True & True & $max(1,1,0)=1$ & $max(1,0,1)=1$ & $max(0,1,1)=1$ \\ 
			\rowcolor{tablealtrowcolor} False & True & True & $max(0,1,0)=1$ & $max(0,0,1)=1$ & $max(1,1,1)=1$ \\ 
			True & False & True & $max(1,0,0)=1$ & $max(1,1,1)=1$ & $max(0,0,1)=1$ \\ 
			\rowcolor{tablealtrowcolor} True & True & False & $max(1,1,1)=1$ & $max(1,0,0)=1$ & $max(0,1,0)=1$ \\ 
			False & False & True & $max(0,0,0)=0$ & $max(0,1,1)=1$ & $max(1,0,1)=1$ \\ 
			\rowcolor{tablealtrowcolor} False & True & False & $max(0,1,1)=1$ & $max(0,0,0)=0$ & $max(1,1,0)=1$ \\ 
			True & False & False & $max(1,0,1)=1$ & $max(1,1,0)=1$ & $max(0,0,0)=0$ \\ 
			\rowcolor{tablealtrowcolor} False & False & False & $max(0,0,1)=1$ & $max(0,1,0)=1$ & $max(1,0,0)=1$ \\
								
				\hline
			\end{colortabular}
		}
	}
		\vspace{-0.3cm}
	\caption{Makespan at vertices $C^{(5)}$, $C^{(6)}$ and $C^{(7)}$ for different truth value assignments to $V_i, V_j$ and $V_k$ in Figure \ref{fig:var-clause}$(b)$.} 
	\vspace{-1cm}
	\label{table:1}
	\end{table}


We also prove hardness of approximation, both for the minimum-makespan problem and for the minimum-resource problem.  We begin with the minimum-makespan problem.

\vspace{-0.1cm}
\begin{theorem}
  \label{thm:hardness-approx-makespan}
	The minimum-makespan discrete resource-time tradeoff problem that allows resources to be reused only over paths cannot have a polynomial-time approximation algorithm with approximation factor less than $2$ unless $P = NP$.
\end{theorem}

\vspace{-0.2cm}
\begin{proof}
	We prove the theorem by contradiction. Let's assume that there is a polynomial time approximation algorithm with factor less than $2$. Given a formula with $n$ variables and $m$ clauses, we construct the reduced DAG as described in the proof of Lemma~\ref{lem:zero-one}. If the formula is a valid 1-in-3SAT instance, then OPT (i.e., the optimal solution) has a makespan of $1$ using $(n+2m)$ units of resource in the reduced DAG. The approximation algorithm will return a schedule with makespan less than $2$ using $(n+2m)$ units of resource. If the formula is not a valid 1-in-3SAT instance, then OPT's makespan is greater than or equal to $2$. So, the approximation algorithm will have a schedule with makespan greater than or equal to $2$. Thus, using a polynomial time algorithm one can solve a strongly NP-hard problem. This is a contradiction. Hence, there exists no polynomial time approximation algorithm for resource-time-reuse-path problem with factor less than $2$ unless $P = NP$. 
\end{proof}

\begin{figure*}[h!]
\begin{minipage}{\textwidth}
\vspace{3cm}
	\centering
	\includegraphics[scale=0.35]{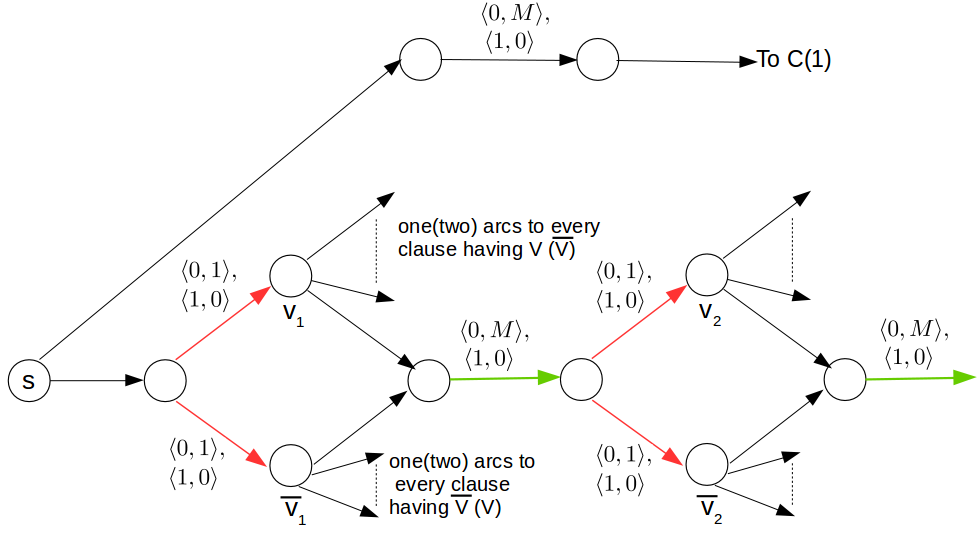}
	\figcaption{The variable gadgets chained together for the hardness of approximation of the minimum-resource problem (Theorem \ref{thm:hardness-approx-resources}).}
	\label{fig:vars}

	\vspace{0.5cm}

	\includegraphics[scale=0.35]{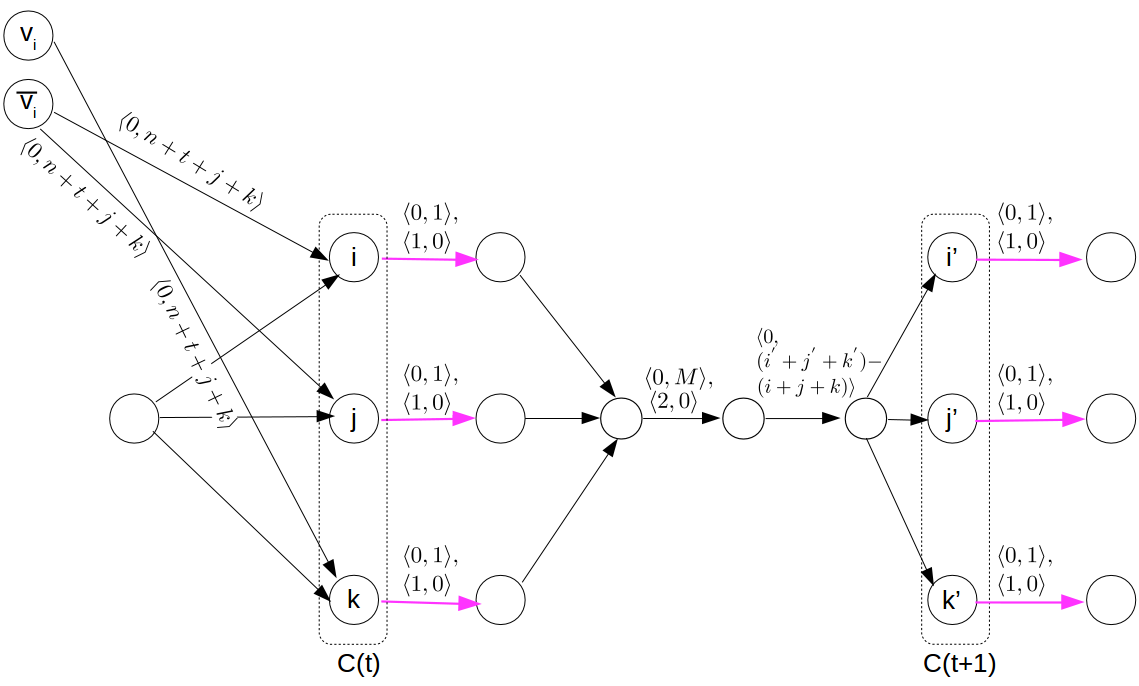}
	\figcaption{The clause gadgets chained together for the hardness of approximation of minimum-resource problem.}
	\label{fig:clauses}

\vspace{1cm}
\end{minipage}
\end{figure*}

Now, we turn attention to the minimum-resource problem:

\begin{theorem}
  \label{thm:hardness-approx-resources}
	The minimum-resource discrete resource-time tradeoff problem that allows resources to be reused only over paths cannot have a polynomial-time approximation algorithm with approximation factor less than $3/2$ unless $P = NP$.
%
\end{theorem}

\begin{proof}
  (Sketch) The proof uses a reduction from 1-in-3SAT; the
  construction is similar to that in the proof of
  Theorem~\ref{thm:strong-hardness}, but has several key differences
  that make it considerably more intricate.

  First, for each variable $x_i$ we have a gadget similar to before (Figure~\ref{fig:var-clause}$(a)$), with the option to send one unit of resource on one of two two-edge paths via a vertex, with the choice of which path indicating whether the variable is set to true or to false.  Unlike the previous construction, we chain the variable gadgets together into a path of gadgets, from a source $s$ to a sink $t$. Refer to Figure~\ref{fig:vars}. A single unit of resource will be moved along the path, using one of each pair of two-edge paths, according to the truth assignments of the variables. A single directed edge, with options $\langle 1,0\rangle$ and $\langle 0,M\rangle$, links variable $x_i$ gadget to variable $x_{i+1}$ gadget.  Node $s$ is connected to the variable $x_1$ gadget with an edge with $\langle 0,0\rangle$.  A property of this construction is that the entry node of the $x_i$ gadget is reached by the unit of resource at exactly time $i-1$, and the exit node of this gadget is reached at time exactly $i$.  At time $n$ the one unit of resource that traverses the path of variable gadgets emerges at time~$n$. Finally, there is also an edge directly from $s$ to $t$ with options $\langle 1,n\rangle$ and $\langle 0, M\rangle$.  In total, two units of resource will be moved through this part of the DAG: one will follow a path through the variable gadgets, according to the truth assignments of the variables, and the other will go directly along the edge $(s,t)$.  Both units of resource will arrive at $t$ at time $n$.

\hide{
  \begin{figure*}[h!]
	\centering
	\includegraphics[scale=0.35]{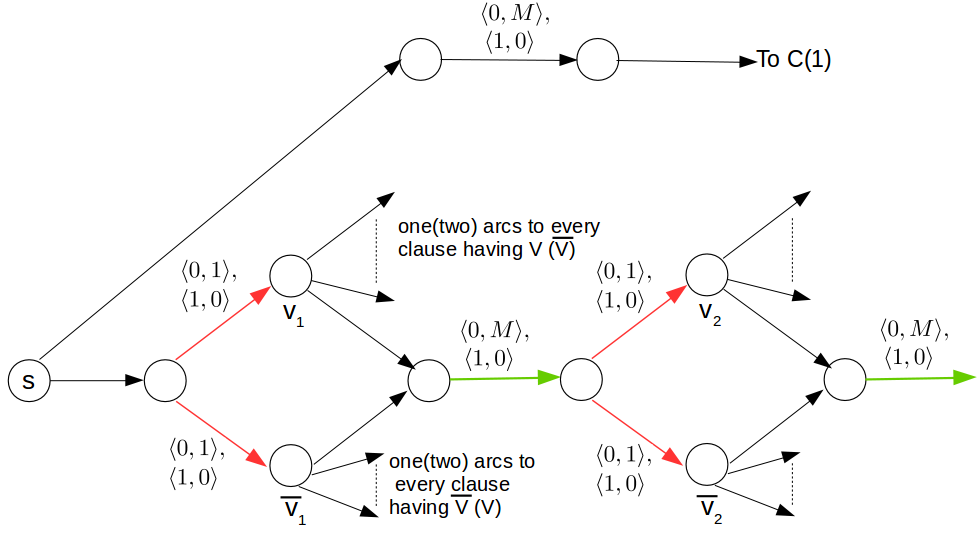}
	\caption{The variable gadgets chained together for the hardness of approximation of the minimum-resource problem (Theorem \ref{thm:hardness-approx-resources}).}
	\label{fig:vars}
\end{figure*}

    \begin{figure*}[h!]
	\centering
	\includegraphics[scale=0.35]{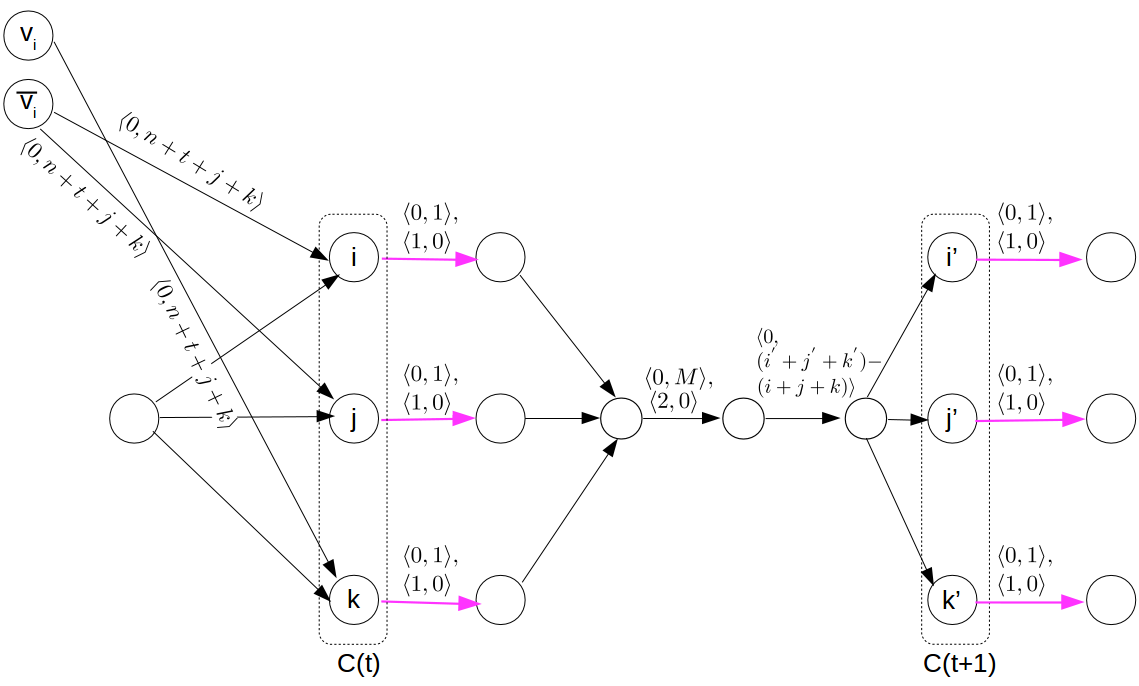}
	\caption{The clause gadgets chained together for the hardness of approximation of minimum-resource problem.}
	\label{fig:clauses}
\end{figure*}
}

  The clause gadget consists of three vertices, each representing a literal. Each clause has an entry vertex and an exit vertex, and they are chained into a path of gadgets, with clauses ordered in a specific way, as described below. Refer to Figure~\ref{fig:clauses}. The exit vertex of one clause has an edge connecting it to the next clause in the order; these edges have specially chosen duration values in order to serve as ``buffers'', as described below. The variable portion of the DAG feeds into the path of clause gadgets, with the 2 units of resource that arrive at $t$ at time $n$ moving along an edge that feeds into the first of the sequence of clause gadgets. Each of the three vertices of a clause gadget corresponds to a literal; each has an input edge coming from one of the two vertices of the variable gadget corresponding to the literal, according to whether the variable appears positively or negatively in the clause. These incoming edges have durations that are carefully chosen, so that the timing is as follows: For a clause with variables $x_i$, $x_j$, and $x_k$, the two units of resource (which came through the variable portion of the DAG before entering the path of clause gadgets) will arrive at the entry to the clause at exactly time $n+i+j+k$.  The incoming edges from variables to the clause literals have durations chosen just so that the precedence constraints are satisfied ``just in time'', for the two units of resource to pass through the clause gadget literals that are {\em not} true (using edges with duration 0, based on the resource of 1), while the one true literal vertex (who was reached within the clause gadget via an edge of duration 1, instead of 0, since there was no resource associated with it) is reached 1 unit of time sooner (from the variable gadget), to compensate.  The net result is that both units of resource emerge out of a clause at time $n+1+i+j+k$, ready to pass into the buffer and the next clause gadget.  The buffers are selected carefully. 

  Then, we claim that we can achieve makespan $A$ 
  using just the 2 units of resource if and only if the variables are assigned to satisfy the 1-in-3SAT. If the variables are assigned in a way that does not yield all clauses to be true, then we will need at least 3 units of resource to achieve the target makespan. Thus, it is NP-hard to distinguish between needing 2 units and needing 3 units of resource. This implies that it is NP-hard to achieve an approximation ratio better than factor 3/2.  
\end{proof}

\begin{figure*}[h!]
\begin{minipage}{\textwidth}
	\centering
	\includegraphics[scale=0.30]{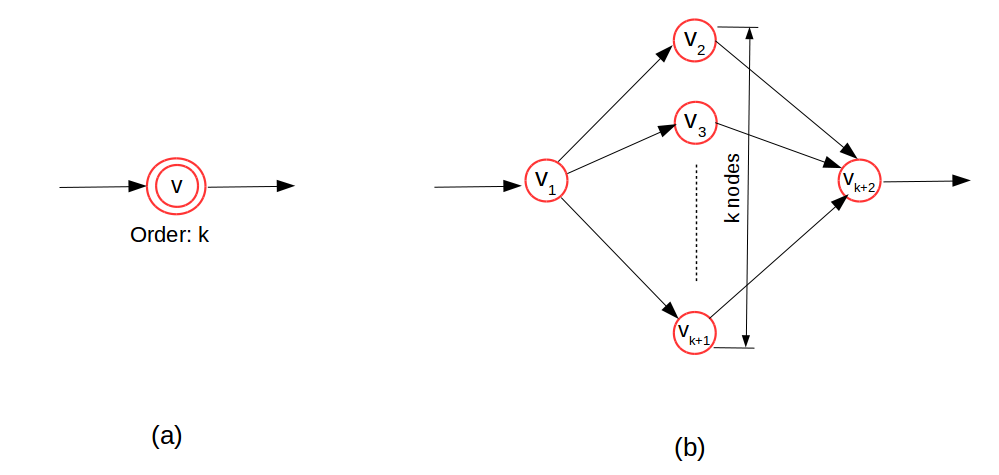}
	\vspace{-0.2cm}
	\figcaption{Composite node (Section \ref{sec:splitting-hardness}).}
	\label{fig:composite}

	\vspace{0.5cm}

	\includegraphics[scale=0.30]{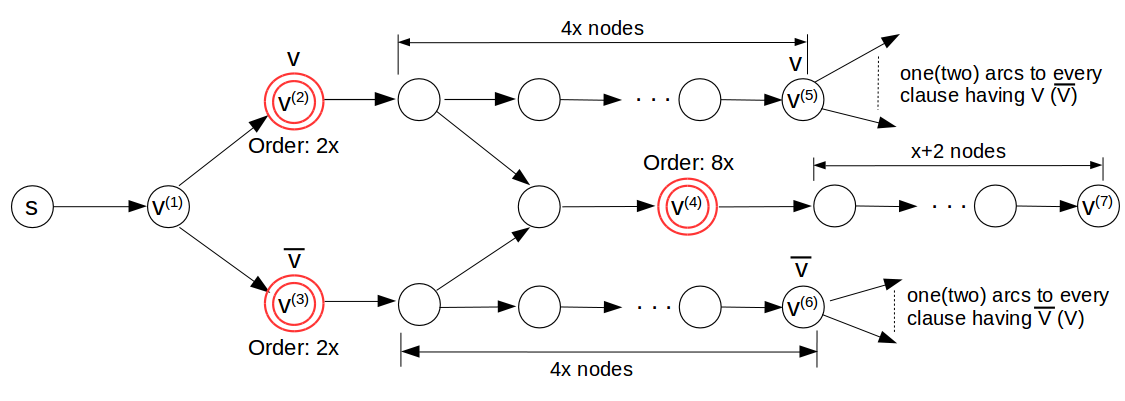}
	\figcaption{Gadget for variable $V$ (Section \ref{sec:splitting-hardness}).}
	\label{fig:binary_var}

	\vspace{0.5cm}

	\includegraphics[scale=0.30]{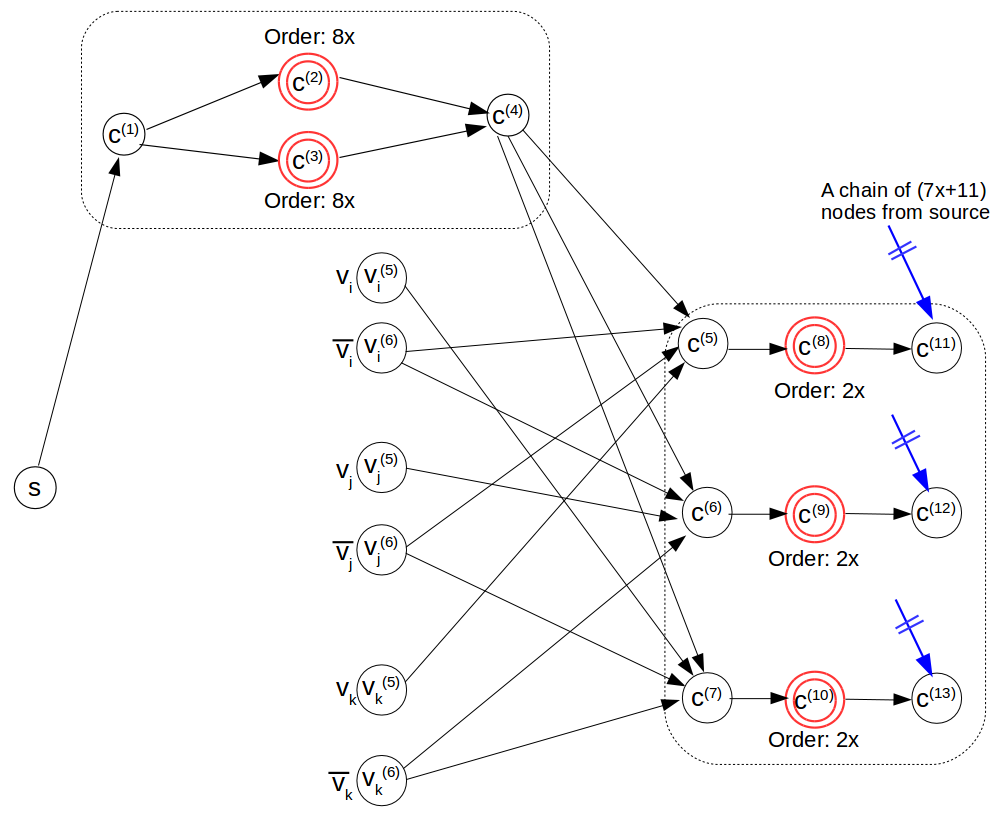}
	\figcaption{Gadget for clause $C=(V_i\lor V_j\lor V_k)$ (Section \ref{sec:splitting-hardness}).}
	\label{fig:binary_clause}
\end{minipage}
\end{figure*}

\vspace{-0.3cm}
\subsection{Reuse Over a Path with Recursive Binary Splitting and $k$-Way Splitting}
\label{sec:splitting-hardness}

We have seen a (strong) NP-hardness proof (Theorem~\ref{thm:strong-hardness}) for the discrete resource-time tradeoff problem with general non-increasing duration functions. In this subsection we strengthen this result by showing that the problem remains hard even when the duration functions arise from recursive binary split reducers and $k$-way split reducers. The proof uses the same general technique as in Section~\ref{general-hard}, but requires more complex gadgets to deal with the restricted duration functions.

\hide{
  \textbf{Makespan of a DAG:} Given a DAG $D$, the cost of a path $\pi$ is the sum of in-degrees of the vertices in path $\pi$. The makespan of $D$ is the maximum cost over all paths in $D$.
We are interested in the problem, given a resource budget $k$, minimize the makespan of $D$. \\
We reduce this problem to the 1-in-3SAT problem which is strongly NP-hard. This will use the same general techniques as in section \ref{general-hard} but will introduce more complex gadgets to deal with the restricted functions (recursive binary splitting and $k$-way splitting).
}

\hide{ 
\begin{figure*}[h!]
	\centering
	\includegraphics[scale=0.35]{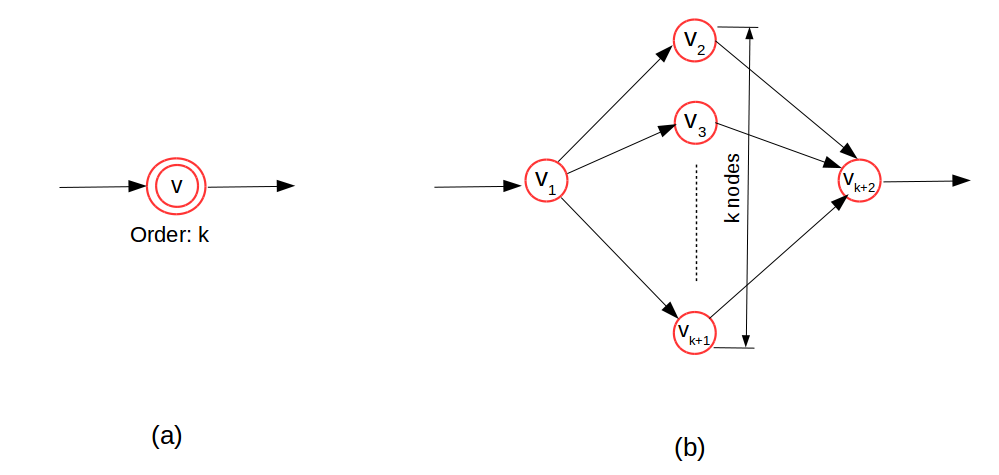}
	\caption{Composite node}
	\label{fig:composite}
\end{figure*}
}

\para{Composite node.} A composite node $v$ of order $k$ is a gadget of $(k+2)$ nodes as shown in Figure \ref{fig:composite}. A composite node can have only one incoming edge and only one outgoing edge. Without using any extra resource, a composite node of order $k$ takes $(k+2)$ units of time to finish its activities. This is because there is one write operation on vertex $v_1$, one write operation on vertex $v_i$ $(2\le i \le k+1)$ and $k$ write operations on vertex $v_{k+2}$.
Using $2$ units of resource with the $k$-way splitting function, all activities can be completed in $(2 + k/2 + 2) = (k/2+4)$ time. Similarly using $2$ units of resource with recursive binary splitting function, all activities will be completed in $(2 + k/2 + \log 2 +1) = (k/2+4)$ time. Thus using $2$ units of resource, composite node $v$ takes $(k/2 + 4)$ units of time using either function.

\hide{
\begin{figure*}[h!]
	\centering
	\includegraphics[scale=0.35]{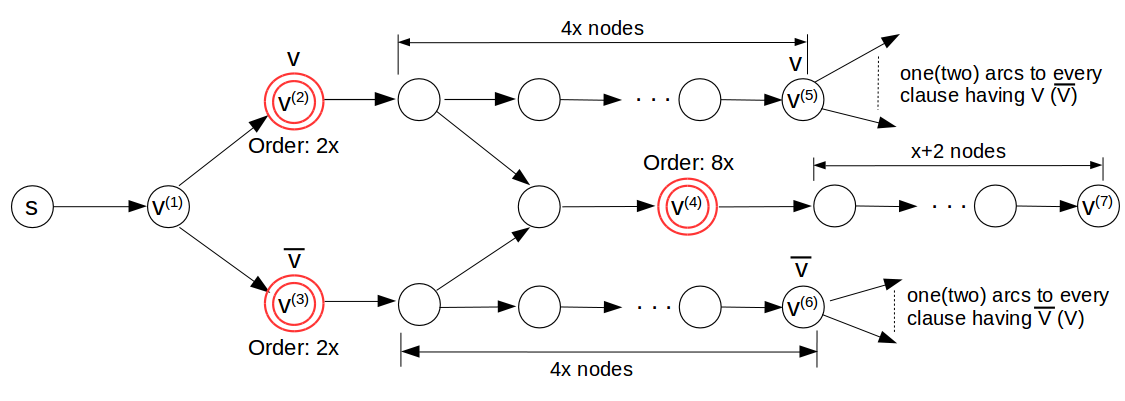}
	\caption{Gadget for variable V}
	\label{fig:binary_var}
\end{figure*}
}

\para{Variable gadget.} The gadget for variable $V$ consists of $3$ composite nodes and other nodes as shown in Figure \ref{fig:binary_var}. Composite nodes $V^{(2)}$ and $V^{(2)}$ are of order $2x$. Composite node $V^{(4)}$ is of order $8x$. There is a chain of $4x$ nodes from $V^{(2)}$ to $V^{(5)}$ inclusive. Similarly there is a chain of $4x$ nodes from $V^{(3)}$ to $V^{(6)}$ inclusive. We will see that unless a variable gadget gets exactly $2$ units of resource, its makespan will be greater than $(7x+2y+12)$ which we will use as the target makespan later in our hardness proof. The values of $x$ and $y$ will be described shortly. Sending $2$ units of resource to node $V^{(2)}$ (Figure \ref{fig:binary_var}) corresponds to setting the variable $V$ to TRUE and sending $2$ units of resources to $V^{(3)}$ corresponds to setting $V$ to FALSE. We will see that sending one unit of resource to $V^{(2)}$ and one unit of resource to $V^{(3)}$ will make the makespan greater than the target makespan.

\para{Clause gadget.} The gadget corresponding to clause $C$ has $13$ vertices $C^{(i)}$ ($1\le i\le 13$) as shown in Figure \ref{fig:binary_clause}. Vertices $C^{(2)}$ and $C^{(3)}$ are composite nodes each of order $8x$. If clause $C$ has three literals $V_i, V_j$ and $V_k$, then vertex $C^{(5)}$ is connected to the vertices $V_{i}^{(6)}, V_{j}^{(6)}$ and $V_{k}^{(5)}$. These vertices correspond to $\neg V_i, \neg V_j$ and $V_k$ respectively. Vertex $C^{(6)}$ is connected to $V_{i}^{(6)}, V_{j}^{(5)}$ and $V_{k}^{(6)}$ . These vertices correspond to $\neg V_i, V_j$ and $\neg V_k$. Vertex $C^{(7)}$ is connected to $V_{i}^{(5)}, V_{j}^{(6)}$ and $V_{k}^{(6)}$. These vertices correspond to $V_i, \neg V_j$ and $\neg V_k$. There are $3$ composite nodes $C^{(8)}, C^{(9)}$ and $C^{(10)}$ each of order $2x$. There is a chain of 
$7x+11$ vertices from $s$ to each vertex in $\left\{ C^{(11)}, C^{(12)}, C^{(13)}\right\}$. We define the ``earliest finish time'' of a node $v$ as the time when all the write operations at $v$ are finished.

In a valid assignment of 1-in-3SAT, we show that for each clause $C$, exactly $2$ vertices of $C^{(5)}, C^{(6)}$ and $C^{(7)}$ will have earliest finish time of $(6x+5)$ and the other one will have earliest finish time of $(5x + 8)$. (Table \ref{binary_table:2})

\para{Value of $x$.} There is only one vertex ($V^{(7)}$) with out-degree zero in every variable gadget $V$. Also, in every clause gadget $C$, there are three vertices $C^{(11)}, C^{(12)}$ and $C^{(13)}$, each with zero out-degree. So, if we connect all such vertices to the sink vertex $t$, then in-degree at $t$ will be $(n + 3m)$.
Let $k$ be the smallest power of $2$ such that $k \ge (n + 3m)$. We perform a recursive binary splitting at vertex $t$. Let $y$ be the height of the binary splitting at $t$ where $y = \log k$. To make $8x > (7x+2y+12)$, we define $x = max\big((2y + 13), 8\big)$.
Hence, the path from any vertex from $\left\{ V^{(7)}, C^{(11)}, C^{(12)}, C^{(13)} \right\}$ to sink $t$ will take time $2y$.

\para{Truth value assignment.} Setting variable $V$ to TRUE implies sending $2$ units of resource through composite vertex $V^{(2)}$. The corresponding earliest finish time at vertex $V^{(5)}$ is $1 + (x+4) + 4x = 5x+5$ and at vertex $V^{(6)}$ is $1 + (2x+2) + 4x = 6x+3$.
Similarly, setting variable $V$ to FALSE implies sending $2$ units of resource through vertex $V^{(3)}$. The corresponding earliest finish time at vertex $V^{(5)}$ is $1 + (2x+2) + 4x = 6x+3$ and at vertex $V^{(6)}$ is $1 + (x+4) + 4x = 5x+5$. 

\hide{
\begin{figure*}[h!]
	\centering
	\includegraphics[scale=0.40]{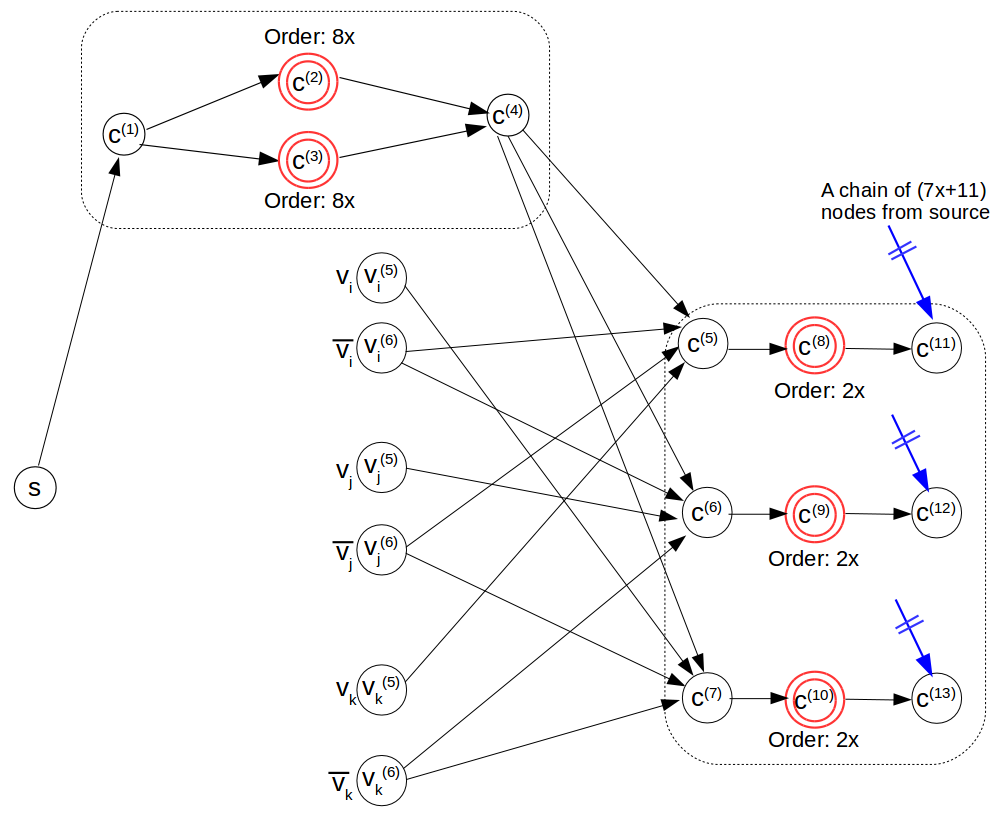}
	\caption{Gadget for clause $C=(V_i\lor V_j\lor V_k)$.}
	\label{fig:binary_clause}
\end{figure*}
}

\begin{lemma}
\label{lm:hardness-k-bin}
	There exists a solution to the input instance of 1-in-3SAT  iff there exists a valid flow of resource through the reduced DAG achieving a makespan of at most $7x + 2y + 12$ using at most $2n + 4m$ units of resource.
\end{lemma}

\begin{proof}
\noindent{\bf Forward direction.} We now prove that if there is a solution to the 1-in-3SAT instance with $n$ variables and $m$ clauses, then the reduced DAG has a makespan of $7x + 2y + 12$ with $2n+4m$ units of resource.
	
	If a variable $V$ is set to TRUE, then we allow $2$ units of resource to flow through vertex $V^{(2)}$ along the path $\langle S, V^{(1)}, V^{(2)}, V^{(4)}\rangle$, otherwise, we allow $2$ units of resource to flow through vertex $V^{3}$ along the path $\langle S, V^{(1)}, V^{(3)}, V^{(4)}\rangle$.
	Assigning TRUE to variable $V$ implies that the earliest finish times at vertex $V^{(5)}$ and $V^{(6)}$ are $5x+5$ and $6x+3$, respectively. Also, the earliest finish time at vertex $V^{(7)}$ is $1+(2+2x)+1+2+(4x+4)+x+2 = 7x+12$. In Figure \ref{fig:binary_clause}, there are $3$ writers from variable gadgets that write on each of the nodes in $\left\{ C^{(5)}, C^{(6)}, C^{(7)} \right\}$. If there are multiple writers ready to write to the same vertex at the same time, we serialize the write operations. For example, if $V_{i} = TRUE, C_{j} = FALSE$ and $V_{k} = FALSE$, then the writer from variable gadget $V_{i}$ is ready to write at time $5x+5$. The writers from $V_{j}$ and $V_{k}$ are ready to write at time $6x+3$. Hence, all three write operations can be completed at time $\max{\left\{ 5x+6, 6x+4, 6x+5 \right\}} = 6x+5$. From Table \ref{binary_table:2}, it is evident that in clause $C$, if only one literal is TRUE and the other two are FALSE, then among $C^{(5)}, C^{(6)}$ and $C^{(7)}$ only one vertex has an earliest finish time of $5x+8$ and the other two have $6x+5$. The vertex with starting time $5x+8$, can finish the activity corresponding to composite node (one of $C^{(8)}, C^{(9)}$ and $C^{(10)}$) of order $2x$, in another $2x+2$ units of time without using any resource. Hence, it will finish at time $5x+8+2x+2= 7x+10$. Each of the other two vertices with earliest finish time of $6x+5$ takes $2$ units of resource flowing from vertex $C^{(4)}$ and finishes the composite node's activity at time $(6x+5)+(x+4) = 7x+9$.
	There is a chain of $7x+11$ nodes from the source vertex to each of the vertices in $\left\{ C^{(11)}, C^{(12)}, C^{(13)} \right\}$. Thus, the earliest finish time at each of those three vertices is $7x+12$. Together, with $2y$ units of time to sink vertex $t$, the total makespan is $7x+2y+12$.  
	
\noindent{\bf Backward direction.} To achieve a makespan of $7x+2y+12$, every variable gadget requires $2$ units of resource and each clause gadget requires $4$, otherwise the makespan will be $8x$ which is larger than $7x+2y+12$ because $x> 2y+12$. Also, any resource used in a variable gadget cannot be used further in any other variable or clause gadget because the resource can be reused over a path only. Similarly, any resource used in any clause gadget cannot be reused in any other gadget. Only one vertex that is either $V^{(5)}$ or $V^{(6)}$, will have the earliest finish time of $5x+5$. Both cannot be $5x+5$, as there is only $2$ units of resource per variable gadget. Both cannot be $6x+3$ as in a clause $C$ where the literal $V$ or $\neg V$ is present, there is an edge from either $V^{(5)}$ or $V^{(6)}$ to each of $C^{(5)}, C^{(6)}$ and $C^{(7)}$. This requires clause gadget $C$ to get $6$ units of resource to achieve a makespan $\leq 7x+2y+12$. But each clause gadget can have exactly $4$ units of resource. Thus, for every variable $V$, for it to be a valid assignment, $V$ is set to either TRUE or FALSE. From Table \ref{binary_table:2}, if a clause has exactly one TRUE literal, then one of the vertices from $C^{(5)}, C^{(6)}$ and $C^{(7)}$ has the earliest finish time of $5x+8$ and the other two have $6x+5$. This requires to have $4$ units of resource to achieve the earliest finish time $\leq 7x+10$ at each of the vertices from $\left\{ C^{(8)}, C^{(9)}, C^{(10)} \right\}$. This can be achieved by assigning $2$ units of resource to those two composite nodes (from $C^{(8)}, C^{(9)}$ and $C^{(10)}$) that start executing at time $6x+5$. The composite node that can start at time $5x+8$ does not use any extra resource. If the clause does not have exactly one TRUE literal, then the clause gadget would require $6$ units of resource to achieve the target makespan. However, we just argued that each clause gadget can have exactly $4$ units of resource. Thus, each clause has exactly one TRUE literal and the 1-in-3SAT instance is also satisfied.
\end{proof}

\begin{table}[!ht]
	\centering
	\scalebox{0.64}[0.7]{
		{\small 
		\hspace{-0.25cm}	\begin{colortabular}{ | c c c c c c |}
				\hline                       
				\rowcolor{tabletitlecolor} $V_i$ & $V_j$ & $V_k$ & $C^{(5)}$ & $C^{(6)}$ & $C^{(7)}$ \\  \hline

			T & T & T & $max( a, a+1, b ) = a+1$ & $max( a, b, a+1 ) = a+1$ & $max( b, a, a+1 ) = a+1$ \\ 
			\rowcolor{tablealtrowcolor} F & T & T & $max( b, a, b+1 ) = a$ & $max( b, b+1, a ) = a$ & $max( a, a+1, a+2 ) = a+2$ \\ 
			T & F & T & $max( a, b, b+1 ) = a$ & $max( a, a+1, a+2 ) = a+2$ & $max( b, b+1, a ) = a$ \\ 
			\rowcolor{tablealtrowcolor} T & T & F & $max( a, a+1, a+2 ) = a+2$ & $max( a, b, b+1 ) = a$ & $max( b, a, b+1 ) = a$ \\ 
			F & F & T & $max( b, b+1, b+2 ) = b+2$ & $max( b, a, a+1 ) = a+1$ & $max( a, b, a+1 ) = a+1$ \\ 
			\rowcolor{tablealtrowcolor} F & T & F & $max( b, a, a+1 ) = a+1$ & $max( b, b+1, b+2 ) = b+2$ & $max( a, a+1, b ) = a+1$ \\ 
			T & F & F & $max( a, b, a+1 ) = a+1$ & $max( a, a+1, b ) = a+1$ & $max( b, b+1, b+2 ) = b+2$ \\ 
			\rowcolor{tablealtrowcolor} F & F & F & $max( b, b+1, a ) = a$ & $max( b, a, b+1 ) = a$ & $max( a, b, b+1 ) = a$ \\								
				\hline
			\end{colortabular}
		}
	}
		\vspace{-0.3cm}
	\caption{Earliest start time at vertices $C^{(5)}, C^{(6)}$ and $C^{(7)}$ for different assignment of truth values of variable $V_i, V_j$ and $V_k$ in Figure~\ref{fig:binary_clause}, where $a = (6x + 4)$ and $b = (5x + 6)$.} 
	\label{binary_table:2}
	\vspace{-1cm}
	\end{table}

%% file: treewidth.tex
\begin{figure*}[h!]
\begin{minipage}{\textwidth}
\vspace{5cm}
	\centering
	\includegraphics[scale=0.4]{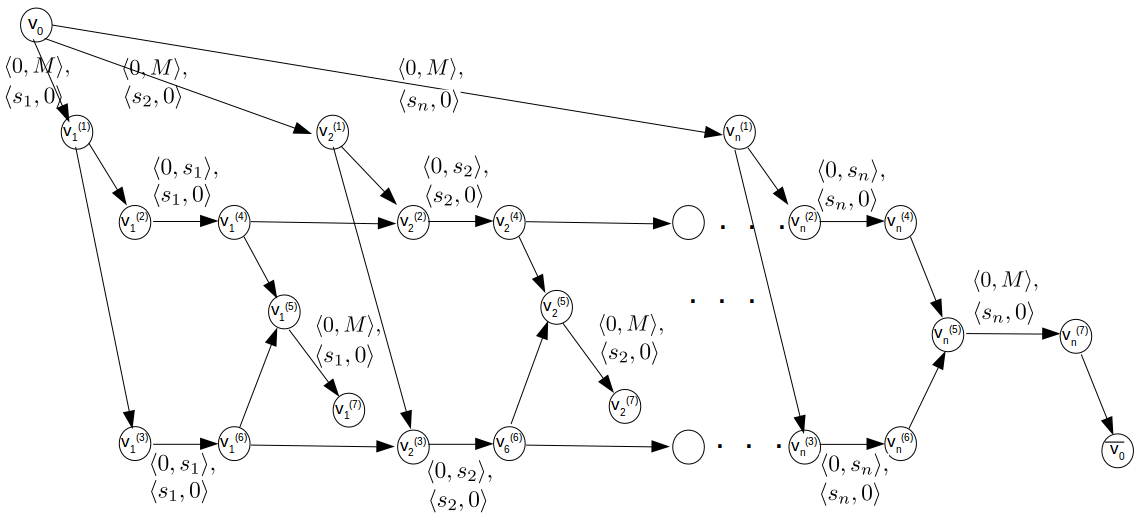}
	\figcaption{Construction for (weak) NP-hardness proof for graphs with bounded treewidth (Section \ref{sec:BoundedTreewidthHardness}).}
	\label{fig:weakly-NP}

	\vspace{0.5cm}

	\includegraphics[scale=0.35]{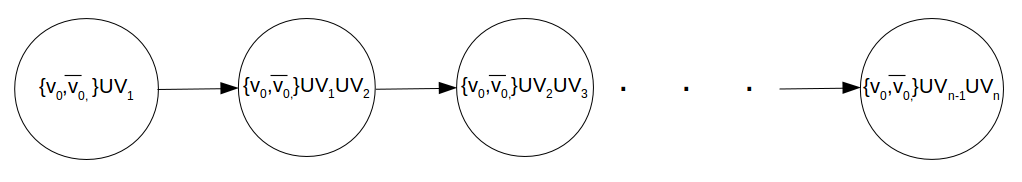}
	\figcaption{Tree decomposition of graph $G$ (Section \ref{sec:BoundedTreewidthHardness}).}
	\label{fig:treewidth}

\vspace{5cm}
\end{minipage}
\end{figure*}

\subsection{Underlying Bounded Treewidth Graph}
\label{sec:BoundedTreewidthHardness}

\hide{
\begin{figure*}[h!]
	\centering
	\includegraphics[scale=0.4]{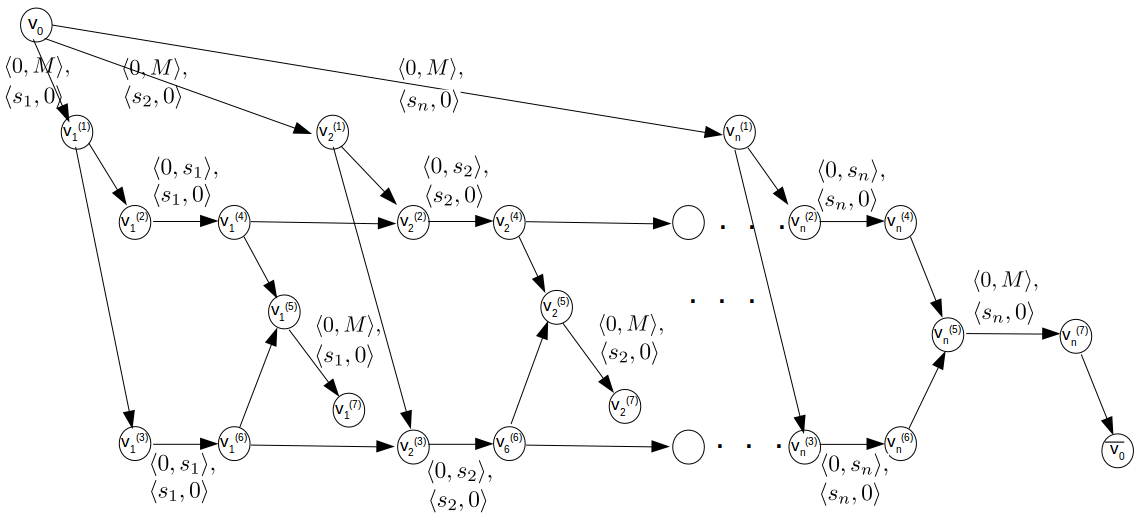}
	\caption{Construction for (weak) NP-hardness proof for graphs with bounded treewidth.}
	\label{fig:weakly-NP}
\end{figure*}

\begin{figure*}[h!]
	\centering
	\includegraphics[scale=0.35]{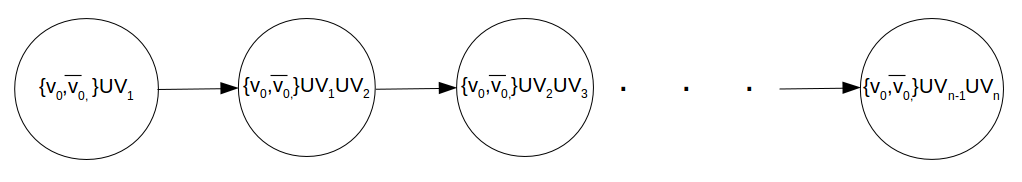}
	\caption{Tree decomposition of graph $G$.}
	\label{fig:treewidth}
\end{figure*}
}

Let $G(D)$ be the undirected graph obtained by ignoring the directedness of the edges of a given DAG $D$.
In the case that $G(D)$ is a graph of bounded treewidth,\footnote{Recall that a tree decomposition of a graph $G=(V,E)$ is a tree $T$ with nodes $X_1, X_2,\ldots, X_n$, $X_i\subseteq V$, satisfying: (1) $\bigcup_i X_i=V$; (2) For edge $(u,v)\in E$ there exists a $X_i$ with $u,v\in X_i$; (3) For any two nodes, $X_i$ and $X_j$, in $T$, if node $X_k$ is in the (unique) path between $X_i$ and $X_j$ in $T$, then $X_i \cap X_j \subseteq X_k$. The {\em width} of the tree decomposition is $\max_i |X_i| - 1$, and the {\em treewidth} of $G$ is the minimum width over all tree decompositions of~$G$.}
we show that the offline minimum-makespan and minimum-resource problems on $D$ are (weakly) NP-hard.
(Note that Theorem~\ref{thm:strong-hardness} proving the strong NP-hardness of the problems does not assume that the underlying undirected graph is of bounded treewidth.)

\begin{theorem}
  It is weakly NP-hard to decide if there exists a solution to the (offline) discrete resource-time tradeoff problem, with resource reuse over paths and a non-increasing duration function, satisfying a resource bound $B$ and a makespan bound $T$, provided the undirected graph obtained by ignoring the directedness of the edges of the input DAG is of bounded treewidth.
\end{theorem}

\hide{
\begin{theorem}
  In the case that the undirected DAG has bounded treewidth, it is weakly NP-hard to decide if there exists a solution to the discrete resource-time tradeoff problem, with reuse and a non-increasing tradeoff function, satisfying a resource bound $B$ and a makespan bound $T$.
\end{theorem}
}

The proof of this theorem is based on a reduction from {\sc Partition}\cite{Garey:1979:CIG:578533}. The
construction is shown in Figure \ref{fig:weakly-NP}. The input instance is a set $S=\{s_1,s_2,\ldots,s_n\}$ of $n$ positive integers; let $B=\sum_{i=1}^n s_i$. The {\sc Partition} problem asks if there is a partition of $S$ into subsets $S_1$ and $S_2$ such that the sums of the values in the two subsets are the same (i.e., exactly $B/2$). In this construction we have a total of $B$ resources to allocate in our program. The value $M$ is chosen to be greater than $B/2$, the target makespan, ensuring that memory resources must be allocated to these nodes. This ensures that at least $s_i$ units of resource pass through each $v_i^{(1)}$, constructing our numbers. From each $v_i^{(1)}$ there are two choices of nodes, $v_i^{(2)}$ and $v_i^{(3)}$, to pass the resources onto each of which will either utilize $s_i$ resources or increase the makespan on that path by $s_i$. The pair also funnel the resources into a sink vertex $\overline{v_0}$ with a potential makespan cost of $M$ which ensures that their resources cannot be passed along to nodes $v_j^{(2)}$ and $v_j^{(3)}$ to the right (i.e., $j > i$). Thus the top and bottom paths represent our two sets and for each $v_i$ we must allocate $s_i$ makespan to either the top or the bottom path. Thus a total makespan of $B/2$ can only be achieved iff there is a partition of the $s_i$'s into two sets such that each set sums to $B/2$.


To see that the constructed graph has bounded treewidth, 
let $V_i = \{v_i^{(j)}\}$, where $1\le j\le 7$. Vertices $v_i^{(7)}$ for $1\le i \le n$ are connected to the sink vertex $\overline{v_0}$.
Then $G$ has a tree decomposition $T$ with nodes $S_i$, $1\le i \le n$, as shown in Figure \ref{fig:treewidth}, with $S_i$ defined as follows:
$S_1 = \{v_0, \overline{v_0}\}\cup V_1$; 
$S_i = \{v_0, \overline{v_0}\}\cup V_{i-1}\cup V_{i}$, for $2\le i \le n$.
We claim that $T$ is a valid tree decomposition. It is evident that $\cup_{1\le i \le n} S_i = V$.
\hide{
There are 4 types of edges: $(v_0, v_i^{(1)})$, $(v_i^{(4)}, v_{i+1}^{(2)})$, $(v_i^{(6)}, v_{i+1}^{(3)})$, and edges having both endpoints in the same set $V_i$, for $1\le i < n$. 
}
%
From the construction of $S_j$ $(1\le j \le n)$, it is clear that, for each edge $(u,v)$ of the graph $G$, there exists a node $S_j$ with $u,v\in S_j$.
For any $S_{i}$ and $S_{j}$, with $j > i+1$ and $1 \le i \le (n-2)$, we have $S_{i} \cap S_{j} = \{v_0, \overline{v_0}\}$, and, for any node $S_k$ ($i < k < j$), on the path between $S_i$ and $S_j$, we have $v_{0}\in S_k$ and $\overline{v_{0}}\in S_k$, so that $S_i\cap S_j\subseteq S_k$.
Thus, $T$ is a valid tree decomposition, and it has width 15 ($\max_i |S_i| -1 = 15$), so the treewidth of $G$ is at most $15$.

%% file: conclusion.tex
\vspace{-0.3cm}
\section{Conclusion}

In this paper we introduce the discrete resource-time tradeoff problem with resource reuse in which each unit of resource is routed along a source to sink path and is possibly used and reused to expedite activities encountered along that path. We consider two different objective functions: (1) optimize makespan given a limited resource budget and (2) optimize resource requirement given a target makespan.

Our original motivation came from a desire to mitigate the cost of data races
%
%
in shared-memory parallel programs by using extra space to reduce the time it takes to perform conflict-free write operations to shared memory locations. We consider three duration functions: general non-increasing function for the general resource-time question, and recursive binary reduction and multiway (k-way) splitting for the space-time case.

We present the first hardness and approximation hardness results as well as the first approximation algorithms for our problems. We show that the makespan optimization problem is strongly NP-hard under all three duration functions. When the duration function is general non-increasing we also show that it is strongly NP-hard to achieve an approximation ratio less than $2$ for the makespan optimization problem and less than $\frac{3}{2}$ for the resource optimization problem. We give a $\left(\frac{1}{\alpha}, \frac{1}{1-\alpha}\right)$ bi-criteria (resource, makespan) approximation algorithm for that same duration function, where $0<\alpha<1$. We present improved approximation ratios for the recursive binary reduction function and the multiway ($k$-way) splitting functions.

\begin{acks}
	\scriptsize{
	This work is supported in part by NSF grants CCF-1439084, CCF-1526406, CNS-1553510, IIS-1546113 and US-Israel Binational Science Foundation grant number 
	2016116. }
\end{acks}

\hide{
We have introduced the problem....

We give the first algorithmic results on this class of problems..
}

%% file: threeDMatching_hardness.tex
\section{Alternate hardness proof from numerical 3D matching}
\label{sec:3d-matching}
We give a polynomial-time reduction from the numerical 3-dimensional matching problem to the discrete resource-time tradeoff problem (with resource reuse over paths and a non-increasing duration function). 

Numerical 3-dimensional matching problem: 
Given $A = $ \\ $\{ a_1, a_2, \cdots a_n \}, B = \{ b_1, b_2, \cdots b_n \},$ and $C = \{ c_1, c_2, \cdots c_n \}$, partition $A\cup B \cup C$ into $n$ triples $S_i \in A\times B\times C$ of equal sum $T = (\sum A + \sum B + \sum C)/n$.

\begin{figure*}[h!]
\begin{minipage}{\textwidth}
\vspace{0.5cm}
	\centering
	\includegraphics[scale=0.33]{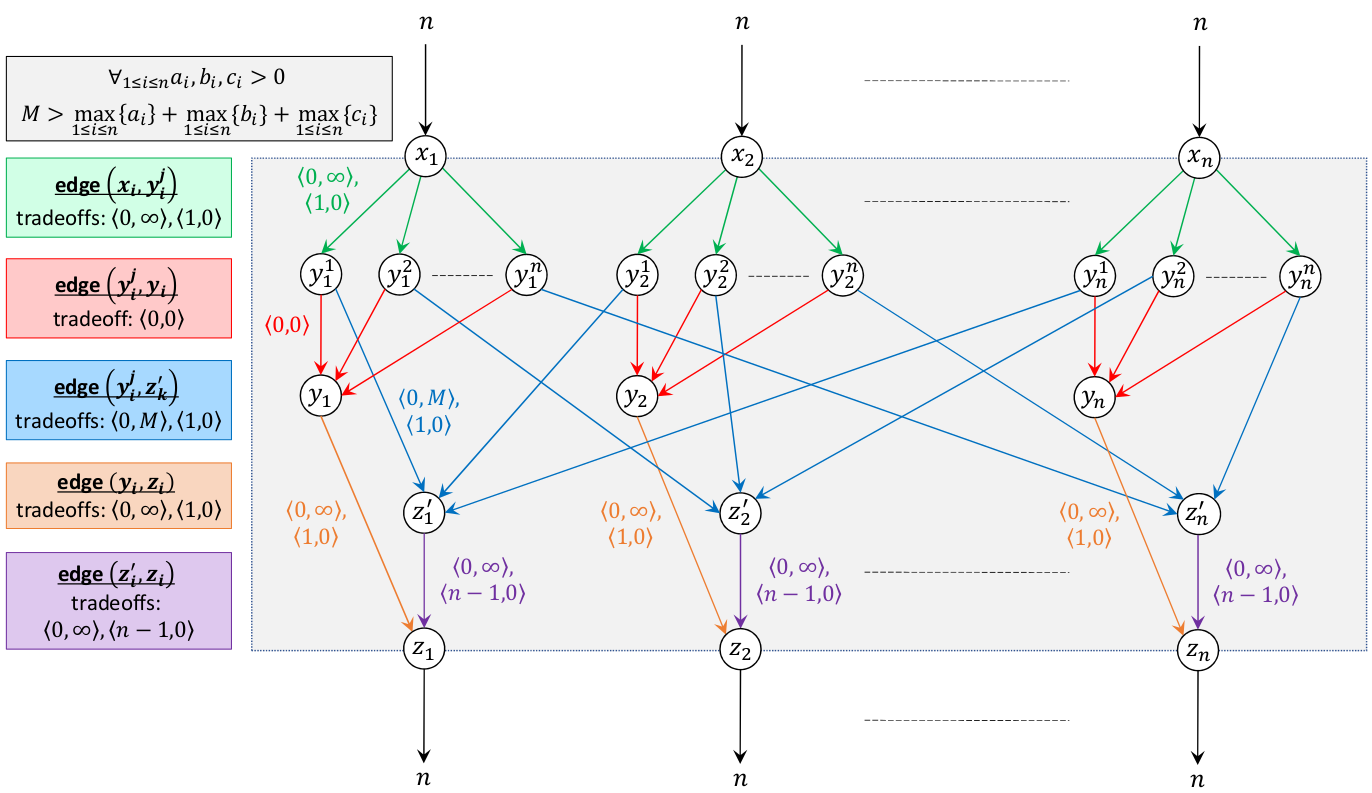}
	\figcaption{Bipartite matcher gadget (Section \ref{sec:3d-matching}).}
	\label{fig:bipartite}

	\vspace{0.5cm}

	\includegraphics[scale=0.33]{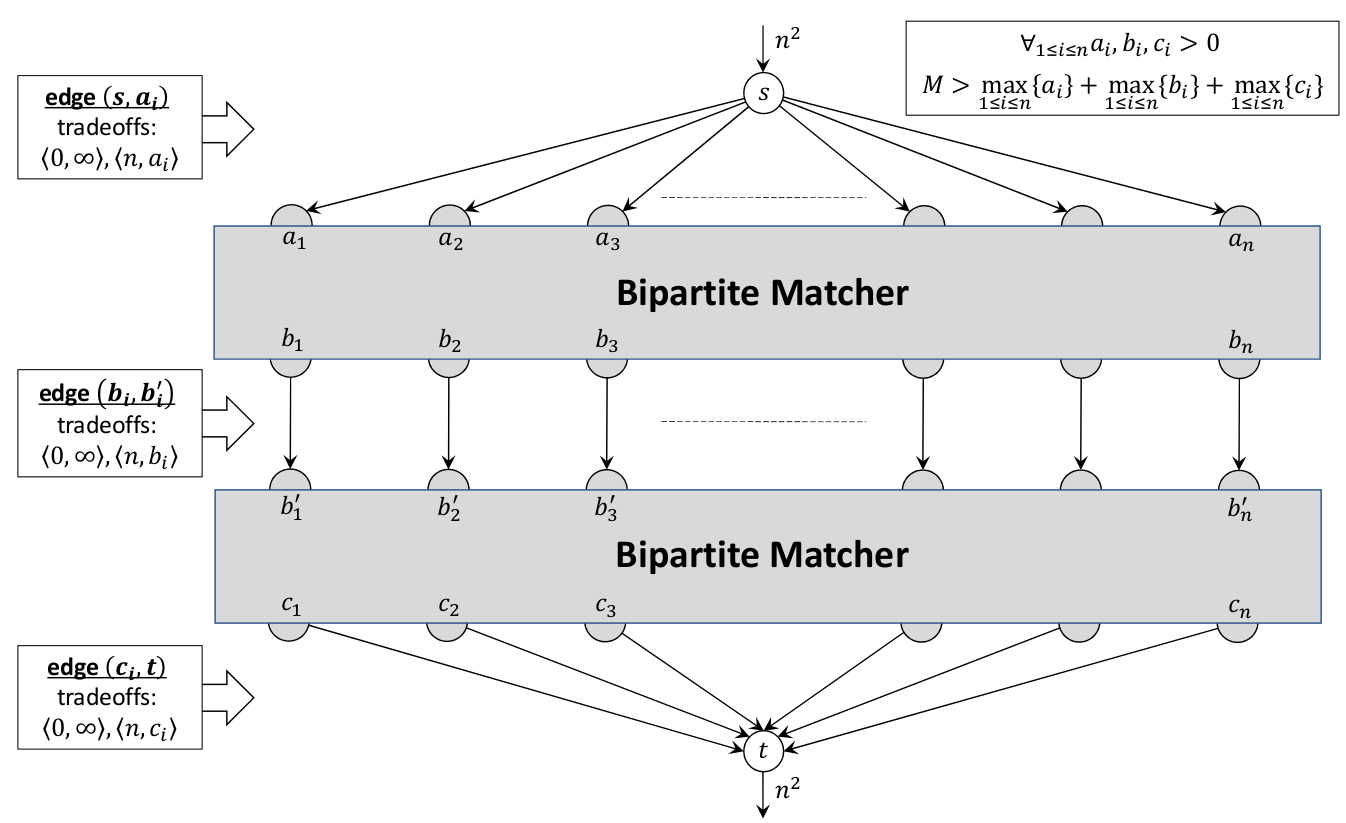}
	\figcaption{Reduced DAG from a numerical 3D matching instance (Section \ref{sec:3d-matching}).}
	\label{fig:numerical-hardness}
\vspace{0.5cm}
\end{minipage}
\end{figure*}

\hide{
\begin{figure*}[h!]
	\centering
	\includegraphics[scale=0.35]{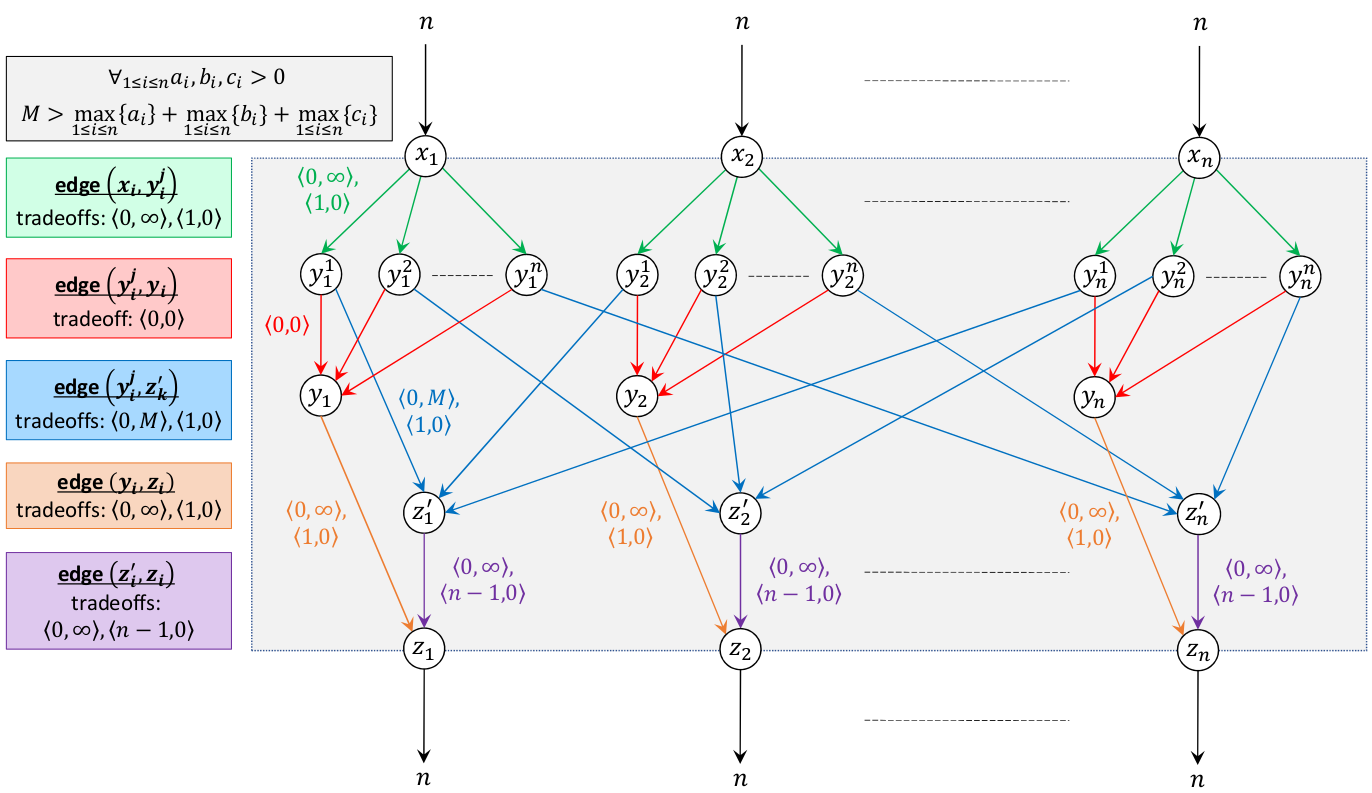}
	\caption{Bipartite matcher gadget}
	\label{fig:bipartite}
\end{figure*}

\begin{figure*}[h!]
	\centering
	\includegraphics[scale=0.35]{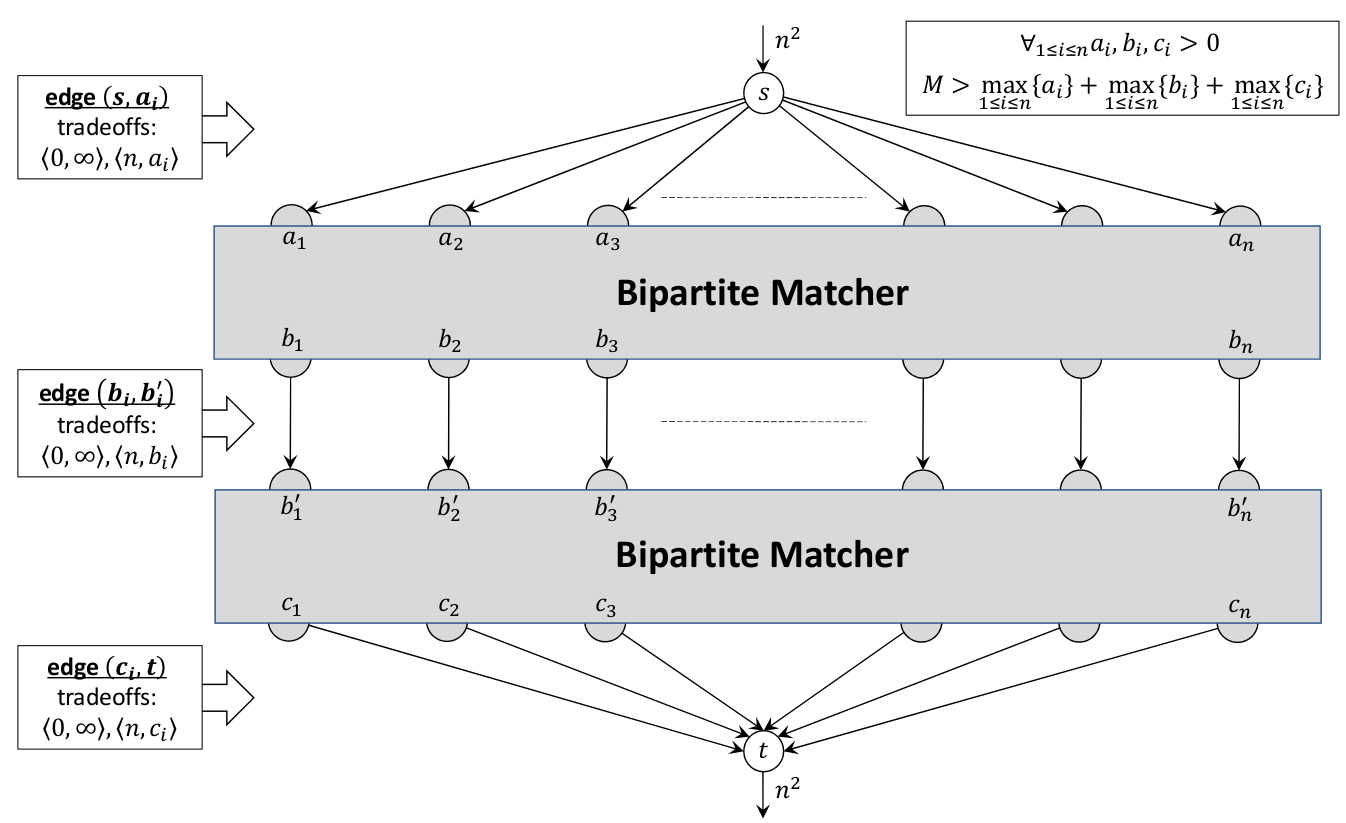}
	\caption{Reduced DAG from a numerical 3D matching instance}
	\label{fig:numerical-hardness}
\end{figure*}
}

Given an instance of the numerical 3D matching problem, we create a DAG $D$ with source $s$ and sink $t$ as shown in Figure \ref{fig:numerical-hardness}. For each $a_i \in A$, there is an edge $(s,a_i)$ in $D$.
The space-time tradeoff function at edge $(s,a_i)$ is $\{\langle 0, \infty \rangle,  \langle n, a_i\rangle\}$. Recall that, this means that with zero resource, it takes infinite time to finish the activity $(s,a_i)$ and with $n$ units of resource it finishes in time $a_i$. We create a gadget that has $n$ incoming edges and $n$ outgoing edges. We call the gadget a {\em bipartite matcher} (Figure \ref{fig:bipartite}) as it matches (a $1:1$ mapping) the incoming edges to the outgoing edges. We describe the bipartite matcher in the next paragraph. For each $b_i \in B$, there is an edge $(b_i,b_i^{'})$ in $D$. The tradeoff function at edge $(b_i,b_i^{'})$ is $\{\langle 0, \infty \rangle,  \langle n, b_i\rangle\}$. We put all the $n$ edges $(b_i,b_i^{'})$ to a bipartite matcher as its incoming edges. For each $c_i \in C$, there is an edge $(c_i,t)$ in $D$. The tradeoff function at edge $(c_i,t)$ is $\{\langle 0, \infty \rangle,  \langle n, c_i\rangle\}$.

\para{The bipartite matcher gadget.} The gadget has $n$ incoming edges at vertices $\{x_1, x_2, \cdots x_n\}$ and $n$ outgoing edges from $\{z_1, z_2, \cdots z_n\}$. It maps the vertices from $\{x_1, x_2, \cdots x_n\}$ to those in $\{z_1, z_2, \cdots z_n\}$. The mapping is one to one. This works as follows. There are $n$ units of incoming resource at each vertices $x_i$. Every outgoing edge $(x_i,y_i^{j})$ from $x_i$ ($1 \le j\le n$) has a tradeoff function $\{\langle 0, \infty \rangle, \langle 1, 0 \rangle\}$. Hence, each of the outgoing edges $(x_i, y_i^{j})$ from $x_i$ gets one unit of resource. The tradeoff function at edge $(y_i, z_i)$ is $\{\langle 0, \infty \rangle, \langle 1, 0 \rangle\}$ which forces  $y_i^{j}$ to send one unit of resource to $y_i$. The tradeoff function at edge $(y_i^{j}, z_j^{'})$ is $\{\langle 0, M \rangle, \langle 1, 0\rangle \}$. Thus, if $y_i^{j}$ sends one unit of resource to $y_i$, it cannot send any resource to $z_j^{'}$ forcing the activity $(y_i^{j}, z_j^{'})$ to take $M$ units of time to finish. Here,  $M > \max_{1\le i \le n}(a_i) + \max_{1\le i \le n}(b_i) + \max_{1\le i \le n}(c_i)$. The tradeoff function at edge $(z_j^{'}, z_j)$ is $\{\langle 0, \infty \rangle, \langle n-1, 0 \rangle\}$. There are $n$ incoming edges $(y_i^{j}, z_j^{'})$ to $z_j^{'}$. Out of these $n$ incoming edges, $(n-1)$ edges flow $n - 1$ units of resource to $z_j^{'}$ which are then used for the activity at $(z_j^{'}, z_j)$. 

We now show the mapping through an example. Suppose $x_1$ is mapped to $z_3$. Then the corresponding flow is as follows: one unit of resource flows from $y_1^3$ to $y_1$. As the total incoming flow of resource at vertex $y_1^3$ is one, no resource flows from $y_1^3$ to $z_3^{'}$. However, one unit of resource flows from each $y_i^3$ except $y_1^3$ to $z_3^{'}$. The earliest start time $(EST)$ along path $\langle x_1, y_1^{3}, z_3^{'}\rangle$ is $EST(x_1) + M$ while that along path $\langle x_i, y_i^{3}, z_3^{'}\rangle$ for $i \neq 1$ is $EST(x_i)$. This makes the earliest start time at $z_3^{'}$, $EST(z_3^{'}) = \max\{EST(x_1) + M, EST(x_i)\} = EST(x_1) + M$. This holds true because $M > \max_{1\le i \le n}(a_i) + \max_{1\le i \le n}(b_i) + \max_{1\le i \le n}(c_i)$. Also, $n-1$ units of resource flow to $z_3^{'}$ and they are used for the activity $(z_3^{'}, z_3)$ to finish in time $0$. Observe that no $y_1^{i}$ except $y_1^{3}$ can send resource to $y_1$. The gadget has a total resource-inflow of $n^2$. Each of $(z_i^{'}, z_i)$ requires $n-1$ units of resource that sums up to $n^2-n$ units of resource. Each of $(y_i,z_i)$ requires one unit of resource, that sum up to $n$ units of resource. If two of $y_1^{i}$ sends a unit of resource each to $y_1$, then the total resource left to be used by all $(z_i^{'}, z_i)$ is at most $n^2 - n - 1$. Thus at least one of $(z_i^{'}, z_i)$ won't get $n - 1$ units of resource and will take infinite time. Hence, mapping $x_i$ to $y_j$ corresponds to flowing one unit of resource from $y_i^{j}$ to $y_i$ and vice-versa; this makes a one-to-one mapping from $\{x_1, x_2, \cdots x_n\}$ to $\{z_1, z_2, \cdots z_n\}$.

\begin{lemma}
	\label{lem:zero-one-3D}
	There exists a solution to a input instance of numerical 3D matching if and only if there exists a valid flow of resource in the DAG such that the makespan is $2M + T$ with resource bound $B=n^2$.
\end{lemma}

\begin{proof}
If there is a solution in the input instance of numerical 3D matching, then there are $n$ sets, each of type $\{a_i, b_j, c_k\}$ such that $a_i + b_j + c_k = T$. We use first bipartite matcher gadgets to map $a_i$ to $b_j$ and the second bipartite matcher to map $b_j^{'}$ to $c_k$. Each bipartite matcher contributes $M$ in the makespan. $(s,a_i), (b_j,b_j^{'})$ and $(c_k,t)$ adds $T$ to the makespan. Thus the makespan is eaxctly $2M+T$.

If the reduced DAG admits a makespan of $2M + T$ using $n^2$ units of resource, then there is also a solution to the input instance of numerical 3D matching. From the construction of bipartite 3D matching, there is a one-to-one mapping from $a_i$ to $b_j$ and from $b_j^{'}$ to $c_k$. As the makespan is $2M + T$ and each bipartite matcher contributes $M$ to the makespan, this gives a solution to numerical 3D matching.  
\end{proof}